\documentclass[10pt,journal,compsoc]{IEEEtran}
\usepackage{amsmath,amssymb,amsfonts}
\usepackage{algorithmic}
\usepackage{graphicx}
\usepackage{textcomp}
\usepackage{xcolor}
\usepackage{array}
\usepackage[ruled, linesnumbered, vlined]{algorithm2e}
\usepackage{textcomp}
\usepackage{multirow}
\usepackage{bm}
\usepackage{booktabs}
\usepackage{ntheorem}
\usepackage{subfigure} 
\usepackage{epstopdf}
\usepackage{cases}
\newenvironment{sequation}{\begin{equation}\small}{\end{equation}}

\usepackage{makecell,multirow,diagbox}
\usepackage{stfloats}
\usepackage{enumitem}[topsep=3pt,itemsep=3pt]
\setlist{nolistsep}
\usepackage{comment}
\usepackage{xcolor}
\usepackage{widetext}
\newenvironment{proof}{{\noindent\it Proof. }}{\hfill $\blacksquare$\par}
\usepackage{enumitem}[topsep=3pt,itemsep=3pt]
\allowdisplaybreaks[4]

\def\BibTeX{{\rm B\kern-.05em{\sc i\kern-.025em b}\kern-.08em
		T\kern-.1667em\lower.7ex\hbox{E}\kern-.125emX}}

\definecolor{color}{rgb}{0.0, 0, 0}
\definecolor{color1}{rgb}{0.0, 0, 0}

\ifCLASSOPTIONcompsoc
\usepackage[nocompress]{cite}
\else
\usepackage{cite}
\fi

\theoremseparator{.}
\newskip\theorempreskipamount
\newskip\theorempostskipamount
\newtheorem{theorem}{Theorem}
\newtheorem{definition}{Definition}
\newtheorem{lemma}{Lemma}
\theorembodyfont{}

\newtheorem{corollary}{Corollary}

\begin{document}

\title{BARGAIN-MATCH: A Game Theoretical Approach for Resource Allocation and Task Offloading in Vehicular Edge Computing Networks}

\author{Zemin~Sun,\IEEEmembership{}
    	Geng~Sun,~\IEEEmembership{Member,~IEEE},
		Yanheng~Liu,\IEEEmembership{}
		Jian~Wang,~\IEEEmembership{Member,~IEEE}, and
		Dongpu~Cao,~\IEEEmembership{Senior~Member,~IEEE}
		\IEEEcompsocitemizethanks{\IEEEcompsocthanksitem Zemin Sun, Geng Sun, Yanheng Liu, and Jian Wang are with the College of Computer Science and Technology, Jilin University, Changchun 130012, China, and Key Laboratory of Symbolic Computation and Knowledge Engineering of Ministry of Education, Jilin University, Changchun 130012, China. E-mail: \{sunzemin, sungeng, yhliu, wangjian591\}@jlu.edu.cn.
		\IEEEcompsocthanksitem  Dongpu Cao is with the School of Vehicle and Mobility, Tsinghua University, Beijing, China. E-mail: dongpu.cao@ieee.org}
	\thanks{Manuscript received; revised }
	\thanks{(\textit{Corresponding author: Geng Sun.)}}
}

\markboth{Journal of \LaTeX\ Class Files,~Vol.~14, No.~8, August~2015}%
{Shell \MakeLowercase{\textit{et al.}}: Bare Demo of IEEEtran.cls for Computer Society Journals}

\IEEEtitleabstractindextext{%
	\begin{abstract}		
		Vehicular edge computing (VEC) is emerging as a promising architecture of vehicular networks (VNs) by deploying the cloud computing resources at the edge of the VNs. {\color{color} However, efficient resource management and task offloading in the VEC network is challenging. In this work, we first present a hierarchical framework that coordinates the heterogeneity among tasks and servers to improve the resource utilization for servers and service satisfaction for vehicles. Moreover, we formulate a joint resource allocation and task offloading problem (JRATOP), aiming to jointly optimize the intra-VEC server resource allocation and inter-VEC server load-balanced offloading by stimulating the horizontal and vertical collaboration among vehicles, VEC servers, and cloud server.} Since the formulated JRATOP is NP-hard, {\color{color} we propose a cooperative resource allocation and task offloading algorithm named BARGAIN-MATCH, which consists of a bargaining-based incentive approach for intra-server resource allocation and a matching method-based horizontal-vertical collaboration approach for inter-server task offloading.} Besides, BARGAIN-MATCH is proved to be stable, weak Pareto optimal, and polynomial complex. Simulation results demonstrate that the proposed approach achieves superior system utility and efficiency compared to the other methods, especially when the system workload is heavy.
	\end{abstract}
	
	}

\maketitle
\IEEEdisplaynontitleabstractindextext
\IEEEpeerreviewmaketitle

%
%

\section{Introduction}
\label{sec_introduction}

\par  \IEEEPARstart{W}{ith} the development of vehicular networks (VNs) and the ever-increasing number of vehicles on the road, various and explosive applications are emerging such as autonomous driving, auto navigation, and augmented reality. These vehicular applications usually require extensive computation resources and low or ultra-low latency. However, fulfilling the computation-intensive and delay-sensitive tasks is challenging due to the limited computation resources of vehicles. To overcome this challenge, {\color{color} mobile edge computing or multi-access edge computing (MEC) \cite{Porambage2018}} is emerging as a promising technology by shifting the cloud computing resources in close proximity to mobile terminals, leading to the new paradigm of vehicular edge computing (VEC) \cite{Porambage2018,Sabella2016,Duan2020}. The VEC migrates the lightweight and ubiquitous resources from cloud servers to the road side units (RSUs) equipped with VEC servers to extend the computation capabilities of the conventional VNs \cite{Dai2019a}. By offloading the tasks to the VEC servers, the communication latency between the vehicles and the cloud server can be reduced, and the computation overloads on vehicles can be relieved. 

{\color{color}
\par However, compared to cloud computing and wireless networks, the VEC network is characterized by the typical features of both edge computing and VNs, i.e., the limited resources of edge servers and highly dynamic of VNs. Therefore, the joint optimization of resource management and task offloading for VEC network is essential, which however confronts several challenges. 

\textit{{First}}, various tasks of vehicles generally arrive dynamically and have stringent requirements on the offloading service. Internally, for a certain VEC server, the limited computational resource of the VEC server and the stringent requirements of requesting vehicles could result in resource competition inside the VEC server, especially in peak hours. Therefore, it is challenging for a VEC server to decide efficient resource allocation policy to satisfy the heterogeneous and stringent requirements of different tasks under the resource constraint. \textit{{Second}}, high mobility of vehicles and random generation of tasks lead to the spatio-temporally uneven distribution of tasks among VEC servers. Externally, the heterogeneous computational capacity among VEC servers or between VEC server and cloud server further incurs load imbalance and resource under-utilization among the servers. For example, some VEC servers could be overloaded or congested while the others could be underloaded or idle when tasks are randomly offloaded. Therefore, the spatiotemporal heterogeneity among tasks and the computational heterogeneity among servers pose a significant challenge in designing the efficient task offloading scheme to provide service satisfaction for vehicles and load balance among servers. \textit{{Third}}, the unique features of VNs, such as the dynamic of channel and mobility of vehicles, add complexity to integrating these features into the optimization for the VEC network.

\par This work presents a cooperative resource allocation and task offloading approach for VEC network to optimize the resource allocation for servers and offloading satisfaction for vehicles. The main contributions are summarized as follows:

\begin{itemize}
	\item  We employ a hierarchical architecture of resource allocation and task offloading for VEC network to coordinate both the space-time-requirement heterogeneity among tasks and the computational heterogeneity among servers. Specifically, the regional software-defined networking (SDN) \cite{lin2021sdvec}, which separates the data and control planes, is integrated for efficient resource allocation and task offloading. Under the coordination of the controller, the intra-server resource allocation and inter-server load-balanced offloading are optimized by stimulating horizontal and vertical collaborations among vehicle, edge, and cloud layers.
	
	\item We formulate a joint resource allocation and task offloading problem (JRATOP) by jointly optimizing the strategies of resource allocation and pricing, and task offloading, with the aim of maximizing the system utility that is theoretically modeled by synthesizing the unique features of VN channels, nonorthogonal multiple access (NOMA) \cite{Dai2018}, mobility of vehicles, spatiotemporal heterogeneity of tasks,  computational heterogeneity of servers, and energy consumption of nodes. 
	
	\item Due to the NP-hardness of JRATOP, we propose a cooperative resource allocation and task offloading algorithm BARGAIN-MATCH that consists of two components. For intra-server resource allocation, a bargaining game-based incentive approach is proposed to stimulate collaboration between the task (of a vehicle) and a VEC/cloud server for resource allocation and pricing. For inter-server task offloading, a many-to-one matching is constructed between tasks and severs to stimulate edge-edge collaboration for horizontal task offloading and edge-cloud collaboration for vertical task offloading.

	\item The performance of BARGAIN-MATCH is verified through theoretical analysis and simulation. Specifically, BARGAIN-MATCH is proved to be stable, weak Pareto optimal, and polynomial complex. Furthermore, simulation results show that BARGAIN-MATCH can achieve superior system utility and system efficiency, especially when the workload is heavy.
\end{itemize}
}

\par The remaining of this paper is organized as follows. Section \ref{sec_related work} reviews the related work. Section \ref{sec_model} presents the models and preliminaries. The problem formulation is given in Section \ref{sec_problemFormulation}. Section \ref{sec_jointOffloading} elaborates the proposed BARGAIN-MATCH. Section \ref{sec_simulation} shows the simulation results and discussions. In Section \ref{sec_discussion}, we extend the investigation to the scenario with the real-world vehicle applications. This work is concluded in Section \ref{sec_conclusion}. Furthermore, for the sake of readability, all notations are listed in Table \ref{tab_notation}.

{\color{color}
\section{Related work}
\label{sec_related work}

\par In recent years, MEC has emerged as a promising paradigm to provide cloud-computing capabilities at the network edges by deploying lightweight MEC servers ubiquitously in close proximity to end users \cite{Porambage2018,Taleb2017}. Delay-sensitive, computation-intensive, and energy-consuming computation tasks can be offloaded to the MEC severs to improve the quality of service for mobile users. The task offloading and/or resource management for MEC networks has attracted increasing research attempts. Several studies \cite{Kuang2019,Bi2020,Yu2021} focus on the joint task offloading and resource allocation in the single (or double)-user and single-server MEC system, which may not be applicable to the real-world scenarios, especially those with heavy or bursty workloads.

\par To overcome the above challenge, more studies focus on the multi-user and multi-server MEC environment regarding the task offloading and resource allocation. For example, in \cite{Liu2019}, a new scheme is proposed to guarantee the efficiency and reliability of mission-critical task offloading and resource allocation by imposing probabilistic and statistical constraints to the task queue length based on extreme value theory. Apostolopoulos et al. \cite{Apostolopoulos2020} present a novel risk-aware data offloading approach where the risk-seeking offloading behavior of users and the resource over-exploitation of MEC servers are jointly considered by using the prospect theory and tragedy of the commons. Zhang et al. \cite{Zhang2020} employs Lyapunov optimization theory to optimize the task offloading and resource allocation in the MEC-based cloud radio access network, aiming at maximizing the network energy efficiency. Tan et al. \cite{Tan2022} construct a two-level framework for energy-efficient task offloading and resource allocation in OFDMA-based MEC networks. However, these studies mainly focus on MEC networks, which may be not applicable for VNs with highly dynamic vehicles and wireless channels, and uneven distributed nodes and workloads.

In recent years, considerable efforts have been made to improve the performance of VEC network, especially focusing on task offloading to mitigate the competition among vehicles \cite{Sun2019,Wang2020}, to balance the workload \cite{Zhang2020a,Dai2019a}, and to explore the available resources of vehicles \cite{qin2022learning, liu2022mobility,Wang2022}. Sun et al. \cite{Sun2019} propose an adaptive learning based task offloading algorithm based on the multi-armed bandit theory to minimize the average offloading delay. Wang et al. \cite{Wang2020} propose a multi-user non-cooperative computation offloading game, where each vehicle decides whether to offload its task to the VEC server according to the traffic density. An SDN-based VEC architecture is introduced in \cite{Zhang2020a} to provide centralized network management to balance the workload of task offloading. Dai et al. \cite{Dai2019a} construct a cooperative task offloading mechanism based on the queuing theory to minimize the task completion delay and balance the workload at edges. Qin et al. \cite{qin2022learning} focus on exploiting vehicles’ idle and redundant resources for energy efficient task offloading in VNs under information uncertainty. Liu et al. \cite{liu2022mobility} propose a
task offloading scheme by exploiting multi-hop vehicle computation resources in VEC networks. Wang et al. \cite{Wang2022} consider the available neighboring VEC clusters and propose an imitation learning-based task scheduling approach to minimize the system energy consumption. However, these studies mainly focus on optimizing the offloading strategies with insufficient consideration for the resource allocation of VEC servers. Besides, most of them consider the single-server scenario \cite{Sun2019,Wang2020,qin2022learning,liu2022mobility, Wang2022}.

\par Considering the limited resources of edge servers, some studies target on resource management for VEC networks, including spectrum sharing \cite{Liang2019}, computation resource allocation \cite{Zhu2022}, mult-dimensional resource management \cite{Peng2020}, and the exploitation of under-utilized vehicular resources \cite{Duan2022,Wang2022}. For example, Peng et al. \cite{Peng2020} employ a deep learning approach to manage the resources of VEC servers for the delay-sensitive applications of vehicles. Zhu et al. \cite{Zhu2022} adopt a Stackelberg game to model the interaction between vehicles and VEC servers to obtain the price and amount of computation resources to be allocated. In \cite{Duan2022}, the authors focus on the power-aware resource management to jointly optimize the resource utilization and energy efficiency of VEC servers. However, these studies do not consider the offloading strategies from the perspective of vehicles.

\par Since the problems of resource allocation and task offloading are coupled with each other, several research efforts have been devoted to joint resource allocation and task offloading, aiming at delay-driven system utility optimization \cite{Dai2019a,Choo2018,zhou2019computation,zhao2019computation}, spectrum efficiency improvement \cite{li2020joint}, or energy efficiency improvement \cite{huang2020vehicle, huang2021revenue}.
For instance, Dai et al. \cite{Dai2019a} propose a task offloading and resource allocation scheme for VEC networks to minimize the task processing delay under the permissible latency constraint. In \cite{Choo2018}, the task offloading is optimized by maximizing the task completion probability, and the resource allocation is determined by performing a mobility-aware greedy algorithm. Zhou et al. \cite{zhou2019computation} propose an incentive mechanism based on contract-match mechanism to leverage the under-utilized computation resources for task offloading of nearby vehicles. In \cite{zhao2019computation}, the optimal offloading decision and the resource allocation are achieved by using the game theory. Li et al. \cite{li2020joint} consider the influence of time-varying channel on the time-varying spectrum efficiency of task offloading, which is solved by using the branch and bound algorithm. Considering the limited energy of vehicles, Huang et al. \cite{huang2020vehicle, huang2021revenue} propose an energy efficiency-driven approach to reduce the energy costs of vehicles under the constraint of computation resources of VEC servers.

\par Although the above-mentioned work has significantly improve the performance of the VEC network, there are still some problems to be addressed. The heterogeneity among tasks and servers, the load imbalance among servers, and the unique features of VN such as the mobility of vehicles, the dynamic of vehicular channel, and the energy consumption of task execution have not been jointly explored for task offloading and resource allocation. Distinguished from the previous works, this work studies cooperative resource allocation and task offloading in the VEC network, where the stringent requirements of tasks, the heterogeneity among tasks and servers, the load imbalance, the unique features of VNs, and the energy limitation of the VEC networks, are jointly considered.
}

{\color{color}

	\begin{table*}[!hbp]
	\vspace{0em}
	\setlength{\abovecaptionskip}{0pt}%
	\setlength{\belowcaptionskip}{0pt}%
	\caption{Summary of notations}
	\label{tab_notation}
	\renewcommand*{\arraystretch}{.09}
	\color{color}
	\begin{center}
		\begin{tabular}{|p{.183\textwidth}|p{.32\textwidth}||p{.1\textwidth}|p{.28\textwidth}|}
			\hline
			\textbf{Symbol}&\textbf{Description}&\textbf{Symbol}&\textbf{Description}\\
			\hline
				$a \in \{0,\mathcal{E},o\}$ & The offloading destination of a task&$\alpha_i$& The effective switched capacitance of vehicle $i$'s/server $j$'s CPU\\ 
			\hline
				$\beta^{L}/\beta^{NL}$ & The path loss exponent for LoS/NLoS 	communication&$B_{i,j}$ & The bandwidth \\
			\hline
				$c_{j,i}(t)\in[c_{j,i}^{\min}(t),c_{j,i}^{\max}(t)]$ & The unit price of resources by server $j$&$c$ & The speed of light\\ 
			\hline
				$\mathcal{C}_{i}(t)$ & The computation resources required by per bit of task&
				$\mathcal{C}_{i}^{\text{\text{\text{req}}}}(t)$ & The amount of computation resources \\ 
			\hline
				$C_i^a(t)$/$C_j^i(t)$ & The normalized cost of vehicle $i$/server $j$& $C_i^{\max}$/$C_j^{\max}$ & The maximum budget of vehicle $i$/server $j$\\
			\hline
				$\Delta c_{j,i}(t)$&Bid-spread ask&	$\Delta d_{i,j}(t)$&The horizontal distance difference between the current and previous time epochs\\
			\hline
				$d_{i,j}(t)$ & The distance between $i$ and $j$&$d_0$ & The reference distance\\
			\hline
				$\mathcal{D}_{i}^{\text{in}}(t)$&The size of the computation task&	$\mathcal{D}_{i}^{\text{out}}(t)$ & The size of the computation result\\
			\hline
				${\delta_{i}^{i}}^{*}(t)$/${\delta_{j}^{i}}^{*}(t)$&The optimal partitions when vehicle $i$ makes a proposal&	${\delta_{i}^{j}}^{*}(t)$/${\delta_{j}^{j}}^{*}(t)$&The optimal partitions when server $j$ makes a proposal\\
			\hline
				$D_j$&The set of less preferred tasks of server $j$&	$E_{i}^{0}(t)$/$E_{j}^{i}(t)$ & The energy consumption of vehicle $i$/server $j$\\
			\hline
			 	$E_i^{\max}$/$E_j^{\max}$ & The energy constraint of vehicle $i$/server $j$&$\epsilon_{i}(t)$/$\epsilon_{j}(t)$&The discount factor of vehicle $i$/server $j$ \\
			\hline
				$\mathcal{E}=\{1,\ldots,E\}$&The set of VEC servers	& $f_c$ & The carrier frequency\\
			\hline
			 	$F(var_1,c_{j,i}(t))$ &The function in Eq. (24) &$f_i(t)$ & The available computation capacity of vehicle $i$\\
			\hline
				$f_i^{\max}$/$f_j^{\max}$&The computation capability of vehicle $i$/server $j$ &$f_{j,i}(t)$/ $f_{o,i}(t)$ & The computation resources provided by VEC server $j$/ cloud server $o$ to vehicle $i$\\
			\hline
				$g_{i,j}^{L}(t)/g_{i,j}^{NL}(t)$ & The channel power gain between vehicle $i$ and VEC server $j$ & 	$g_{i,j}(t)$/$g_{I,j}(t)$ & The channel power gain between vehicle $i$/$I$ and VEC server $j$\\
			\hline
				$h_{i,j}^{L}(t)/h_{i,j}^{NL}(t)$ & The small-scale fading for LoS/NLoS communication& $i\in \mathcal{V}$& The index of vehicle $i$\\
			\hline
				$j\in \mathcal{E}$ / $j\in\ \{\mathcal{E},o\}$&The index of server $j$ & $\chi_{\sigma}^{L}/\chi_{\sigma}^{NL}$ & The shadowing for LoS/NLoS communication\\
			\hline
				$L_{i,j}^{L}(t)/L_{i,j}^{NL}(t)$ & The pathloss for LoS/NLoS communication &
				$m^{L}/m^{NL}$ & The fading parameter for LoS/NLoS communication\\
			\hline
				 $N_i^{\text{core}}$/$N_j^{\text{core}}$ & The number of vehicle $i$/server $j$'s CPU cores &$N_0$ & The noise power\\
		     \hline
		   	   $\Omega_k^j$/$\Omega_j^k$&The preference of task $k$/server $j$ towards server $k$/task $k$&$o$& Cloud server \\
		     \hline
		     	$\mathbf{P}_i(t)=(X_i(t), Y_i(t), 0)$&The position of vehicle $i$&$p_i^{\text{gen}}(t)$&The indicator of vehicle $i$'s task generation\\
		     \hline 
		     	$\mathbf{P}_j=(X_j, Y_j, 0)$&The position of VEC server $j$&$P_{i}(t)$/$P_{I}(t)$ & The transmit power of vehicle $i$/$I$\\
		     \hline
		    	 $p_{i,j}^{L}$ & The probability of LoS transmission between vehicle $i$ and VEC server $j$	&$\Phi$&Matching result\\
		     \hline
		     	$(\mathcal{P}_B, \mathcal{S}_B, \mathcal{U}_B)$&The triplet of bargaining&$ (\mathcal{P}_M,\Omega, \Phi)$&The triplet of matching \\
		     \hline
		     	$r_c$&The data rate between edge and cloud &$r_{f}$ & The data rate of fiber link\\
		     \hline
		    	$\rho_{k}(j)$/$\rho_{j}(k)$&The preference of task $k$/server $j$ on server $j$/task $k$&$\Phi(k)$/$\Phi(j^{\prime})$&The matching list of task $k$/server $j^{\prime}$ \\
		     \hline
		     	$R_j (j\in\mathcal{E})$&The coverage radius of VEC server $j$& 	$\mathcal{R}_j^i(t)$ & The normalized revenue of server $j$ \\
		     \hline
		   		  $r_{i,j}(t)$ & The data transmission rate between vehicle $i$ and VEC server $j$	&$SW$&Social welfare\\
		     \hline
		     	 $s_i^a(t)$ & The offloading strategy of task $\mathcal{T}_{i}$&
		    	 $\sigma^{L}/\sigma^{NL}$ & The standard deviation of shadowing for LoS/NLoS transmission\\
		    \hline
		    	$s_{j,i}^*(t)=(f_{j,i}^*(t), c_{j,i}^*(t))$&The optimal resource allocation and pricing& ${S}^*_{\text{off}}(t)$/ ${S}^*_{\text{all}}(t)$&The optimal strategy of offloading/resource allocation \\
		    \hline
				$t\in \mathbf{T}=\{0,\ldots,T-1\}$&System timeline&$\Delta t$&Each slot duration\\
			\hline
				$\mathbf{T}(t_0)$&Time epoch&$T_0$& Each time epoch duration\\
			\hline
		   		  $T_{i,j^{}}^{\text{soj}}$&The sojourn time of vehicle $i$ in the coverage of VEC server $j$	&$\tau>0$ & CPU parameter\\
			\hline
				$ \mathcal{T}_{i}(t)$&The task generated by vehicle $i$ at time $t$& $T_{i}^{\max}(t)$ & The maximum tolerable delay of
				task $i$\\
			\hline
		   		 $T_{i}^{0}(t)$/$T_{i}^{j}(t)$/$T_{i}^{o}(t)$& The total delay for offloading task $\mathcal{T}_i(t)$ locally/on edge server $j$/on cloud server $o$&$T_{i,j^{\text{cur}}}^{\text{tran}}(t)$ & The delay for vehicle $i$ to transmit the task to VEC $j^{\text{cur}}$ \\
		   \hline
		  		  $T_{i,j}^{\text{comp}}(t)$/$T_{i,o}^{\text{comp}}(t)$ & The computation delay that the task is processed by VEC server $j$/cloud server $o$& $\tau>0$ & CPU parameter\\
			\hline
				$T_{j^{\text{cur}},j}^{\text{hand}}(t)$/$T_{j^{\text{cur}},o}^{\text{hand}}(t)$ & The horizontal/vertical task handover delay &$\mathcal{T}^{\text{rej}}(t)$&The tasks that are rejected\\
			\hline
		 		 $T_{j,j^{\text{arr}}}^{\text{hand}}(t)$/$T_{o,j^\text{arr}}^{\text{hand}}(t)$ & The horizontal/vertical result handover delay&$\Psi_i^a(t)$ & The normalized satisfaction level of task completion delay\\
			\hline
				$\mathcal{V}_{j}=\{1,\ldots,V_j\}$ & The set of interference vehicles within the range of VEC server $j$& $U_i^a(t)$/$U_j^i(t)$ & The utility obtained by vehicle $i$/server $j$from offloading task $\mathcal{T}_i(t)$\\
			\hline
				$\mathcal{V}=\{1,\ldots,V\}$&The set of vehicles&$\mathcal{V}^{req}(t)$&The requesting vehicle set	\\
		   \hline
 			$v_i(t)\in[v^{\min}, v^{\max}]$&The velocity  of vehicle $i\in \mathcal{V}$& $w_i$/$w_j$ & The weight coefficient of vehicle $i$/server $j$\\
 			\hline 
 			$\zeta_i(t)$&The indicator for the movement direction of vehicle $i$&$P(\zeta_i(t))$& The general model for the movement of vehicles\\
           \hline
		\end{tabular}
	\end{center}
\end{table*}     
}

%
%
\section{Models and Preliminaries}
\label{sec_model}

In this section, a VEC-enabled VN architecture is first introduced, followed by the traffic and mobility model, communication model, computation model, and energy consumption model.

\subsection{System Model}
\label{sec_system_model}


\subsubsection{System Overview}
\label{sec:system_overview}

\begin{figure*}[!hbt] 
	\setlength{\abovecaptionskip}{0pt}   
	\setlength{\belowcaptionskip}{0pt} 
	\centering
	\includegraphics[width =6.6in]{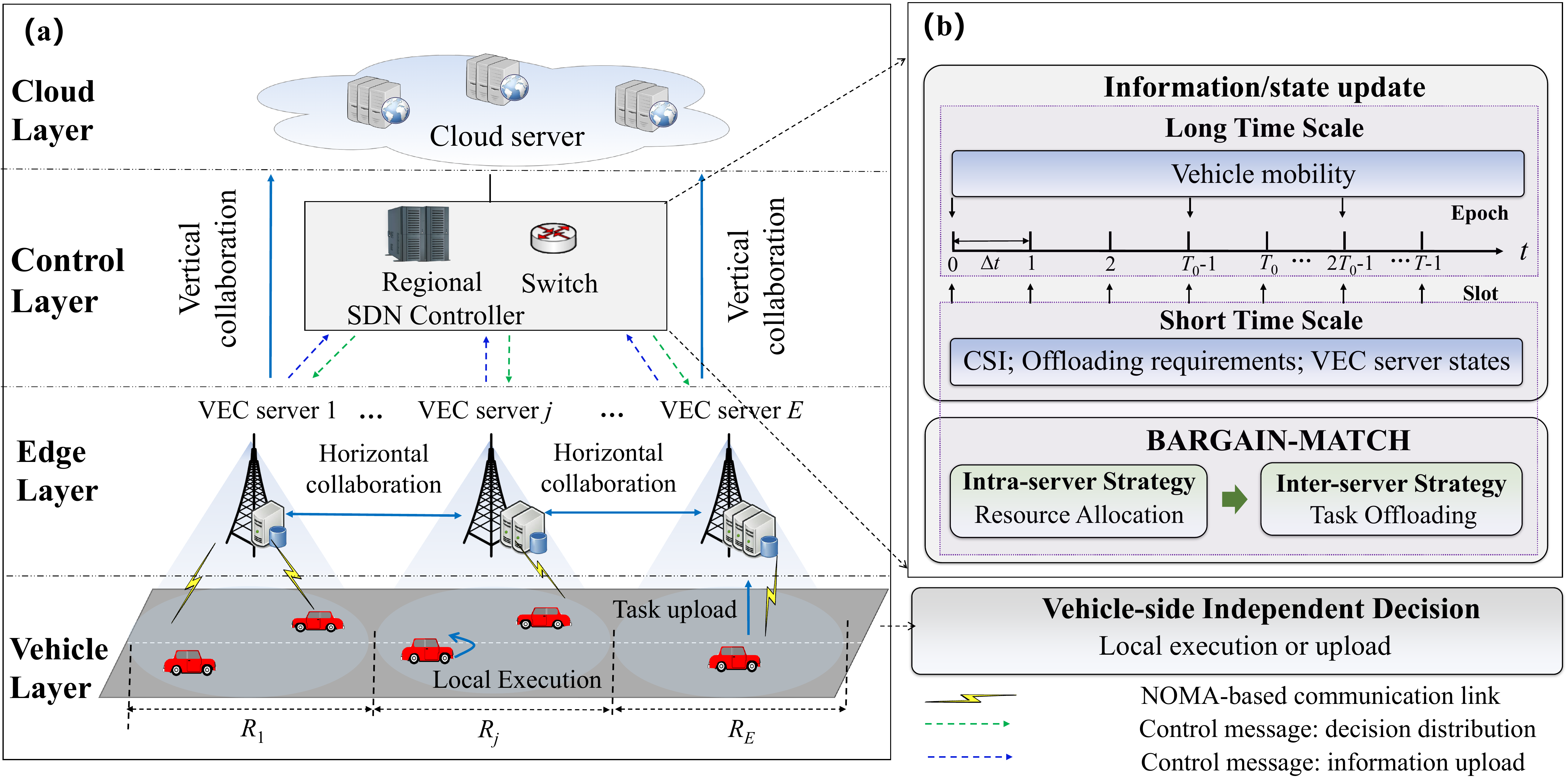}
	\caption{\textcolor{color}{The cooperative architecture of resource allocation and task offloading in VEC network.}}
	\label{fig_gameModel}
	\vspace{-1.2em}
\end{figure*}

{\color{color}
\par  As shown in Fig. \ref{fig_gameModel}, in the spatial domain, we consider an SDN-based hierarchical architecture for VEC network that consists of a vehicle layer with $V$ vehicles \textcolor{color1}{$\mathcal{V}=\{1, \ldots, i, \ldots, V\}$}, an edge layer with $E$ VEC servers \textcolor{color1}{$\mathcal{E}=\{1, \ldots, j, \ldots,E\}$ that are mounted on RSUs \footnote{In this work, the RSU and VEC server are used interchangeably unless otherwise indicated.}}, a control layer with a regional SDN controller, and a remote cloud layer with a cloud server $o$. {\color{color1} For simplicity but without loss of generality, we assume that the vehicles and VEC-mounted RSUs in the system model are equipped with single antenna \cite{Xu2020,Luo2021}.} 
\textit{At the vehicle layer}, vehicles distribute randomly on the multi-lane bi-directional road, run in either direction, and could generate tasks at any time. {\color{color} Each vehicle independently decides whether to process the task locally or upload it to the connected VEC server} \footnote{A vehicle directly uploads the tasks to the VEC server within the current service area since it is unaware of the servers out of the service area \cite{Wang2022}.} by using NOMA. \textit{At the edge layer}, VEC servers are mounted on RSUs that are deployed alongside the road with non-overlapping coverage radius $R_j, j\in \mathcal{E}$, and the service area of each VEC server is defined as the wireless coverage of the RSU. Moreover, the VEC servers are connected to the regional SDN controller \cite{Chen2018,Wu2020} and to each other via high-speed fiber links \cite{Zhang2020,Qian2019, Wang2022} and mobile access gateways (MAGs) \cite{zhou2019energy}. Besides, each VEC server is responsible to collect the local information on its own states, vehicles' states, and the channel state information (CSI), and uploads the information to the control layer. \textit{At the control layer}, the regional SDN controller, on which our algorithm runs, is responsible for decision making and handover management with the support of MAG. For decision making, the following cooperative decisions are performed by our algorithm under the coordination of the controller, i.e., intra-VEC server resource allocation and inter-server task offloading that includes the horizontal task migration among servers and vertical task migration from the VEC server and the cloud server. For handover management, it is incorporated into the task migration to guarantee the connectivity for moving vehicles \cite{zhou2019energy,Malik2021}. With the SDN technology, the controller only needs the knowledge acquired from the edge layer, and makes accordingly decisions by performing the proposed algorithm. \textit{At the cloud layer}, the tasks can be vertically offloaded from the edge layer to the cloud layer under the guidance of the controller.}

{\color{color}
\par In the temporal domain, the VEC system operates in a discrete time-slotted manner. Specifically, the system timeline is equally divided into $T$ time slots, i.e., $t\in \mathbf{T}=\{0,\ldots,T-1\}$, where each slot duration $\Delta t$ is consistent with the coherence block of the wireless channel \cite{Liao2021}. Every $T_0$ consecutive slots are grouped into a time epoch $t_0\in\mathbb{Z}^+$, where $\mathbf{T}(t_0)=\{(t_0-1)T_0,\ldots,t_0\cdot T_0-1\}$. Moreover, the CSI and states of VEC servers and vehicles are captured and updated in different time scales. Specifically, in the short timescale, the information on CSI, offloading requirements of vehicles, and states of VEC servers are captured and updated. In the long timescale, the mobility of vehicles is captured and updated \footnote{The position and velocity of vehicles can be estimated from the GPS data that is known to the RSUs \cite{zhou2019computation}.}. This is because CSI keeps constant in a time slot and varies across different time slots, while the mobility of a vehicle keeps approximately constant during a few consecutive slots \cite{Liao2021}.
}

\subsubsection{Basic Models}
	
\par The basic models for the nodes in the system is illustrated in the following.

{\color{color}
\par \textit{\textbf{Mobility Model.}} The mobility of vehicle $i\in \mathcal{V}$ is denoted as $\left(\mathbf{P}_i(t),v_i(t)\right)$, where $\mathbf{P}_i(t)=\left(X_i(t), Y_i(t), 0\right)$ and  $v^{\min}\leq v_i(t)\leq v^{\max}$ denote the position and velocity  of vehicle $i$ at time $t$, respectively. Suppose that vehicle $i$ is located in the service range of VEC server $j$ at time $t$, the moving direction of vehicle $i$ is denoted as $\zeta_i(t)$ is the indicator function of vehicle $i$'s movement direction towards VEC server $j$. $\zeta_i(t)$ can be estimated based on the difference of the horizontal distance (between $i$ and $j$) between the current time epoch and the previous time epoch, i.e., $\Delta d_{i,j}(t) = |X_i(\left \lfloor{t/T_0}\right \rfloor)-X_j| - |X_i\left(\left \lfloor{t/T_0}\right \rfloor-1 \right)-X_j^{})|$. According to \cite{kim2022modems}, we consider the following two cases. \textbf{i)} If the previous location of the vehicle is known, $\zeta_i(t)$ can be given as an integral indicator, i.e., $\zeta_i(t)=\begin{cases} 1, \ \Delta d_{i,j}(t)>0 \\ -1, \ \Delta d_{i,j^{\text{cur}}}(t)<0\end{cases}$ to indicate the vehicle is moving toward ($\zeta_i(t)=1$) or moving away from ($\zeta_i(t)=-1$) VEC server $j$. \textbf{ii)} If vehicle's mobility is not fully known and should be predicted for the future time, i.e., $t=0$ or $\Delta d_{i,j}(t)=0$, $\zeta_i(t)=\pm P(\zeta_i(t))$ is defined as a continuous variable in [-1,1] to indicate the probability of vehicle's movement direction, where $P(\zeta_i(t))$ is a general model to represent the movement of vehicles which can be set as a typical mobility model, e.g., Markovian mobility \cite{wang2016dynamic}. Therefore, $\zeta_i(t)$ can be concluded as:
\begin{equation}
	\label{eq_mobility_direction}
	\zeta_i(t)=
	\begin{cases}
		1, \qquad \qquad  t>0, \quad \Delta d_{i,j^{\text{cur}}}(t)>0, \\ 
		-1, \qquad \quad \ t>0, \quad \Delta d_{i,j^{\text{cur}}}(t)<0,  \\
		\pm P(\zeta_i(t)), \ t=0 \ \text{or} \ \Delta 	d_{i,j^{\text{cur}}}(t)=0.
	\end{cases}
\end{equation}

\par Given that vehicle $i$ is located within the coverage of VEC server $j^{\text{cur}}$ currently, the remaining sojourn time of vehicle $i$ before leaving the coverage of VEC server $j$ can be obtained as:
\begin{equation}
	\label{eq_sojourn_delay}
	T_{i,j^{}}^{\text{soj}}=\frac{R_j+ \zeta_i(t)\cdot|X_i(t)-X_j^{\text{cur}}|}{v_i(t)}.
\end{equation}

\noindent where $R_{j}^{}$ is the communication radius of VEC server $j^{}$, $|X_i(t)-X_j^{}|$ is the horizontal distance between vehicle $i$ and VEC server $j^{}$.

\par Furthermore, the VEC server $j^{\text{arr}}$ that vehicle $i$ will be attached after time $T_i^{\text{move}}$ can be estimated as:
\begin{sequation}
	\label{eq_arrivalVEC}
	j^{\text{arr}}=j+\zeta_i(t)\cdot \left \lceil \frac{v_i(t)\cdot T_i^{\text{move}}-\left(R_j+ \zeta_i(t)\cdot|X_i(t)-X_j^{}|\right)}{v_i(t)}\right\rceil.
\end{sequation}
}

\par\textit{\textbf{Vehicle Model.}} Each vehicle $i \in \mathcal{V}$ is characterized by $\left(\mathbf{P}_i(t), v_i(t), f_i^{\max}, p_i^{\text{gen}}(t)\right)$, where $\mathbf{P}_i(t)=\left(X_i(t), Y_i(t), 0\right)$ and $v_i(t)$ denote the position and the velocity of vehicle $i$ at time $t$, respectively, and $f_i^{\max}$ is the computation capability of vehicle $i$. {\color{color} Moreover, each vehicle can generate multiple tasks during the system timeline, where in each time slot the vehicle could generate one task or not. The task generation of vehicle $i$ in time slot $t$ is denoted by a binary indicator $p_i^{\text{gen}}(t)\in\{0,1\}$, where $p_i^{\text{gen}}(t)=1$ means that vehicle $i$ generates a task.} Considering the resource limitation, each vehicle is assumed to be equipped with one CPU core \cite{Dai2019a}.

{\color{color}
\par\textit{\textbf{Task Model.}} The task generated by vehicle $i$ in time slot $t$ is denoted as $ \mathcal{T}_{i}(t)=\left(\mathcal{D}_{i}^{\text{in}}(t), \mathcal{D}_{i}^{\text{out}}(t),\mathcal{C}_{i}(t), \mathcal{C}_{i}^{\text{\text{\text{req}}}}(t), T_{i}^{\max}(t)\right)$, where $\mathcal{D}_{i}^{\text{in}}(t)$ is the size of the input computation task, $\mathcal{D}_{i}^{\text{out}}(t)$ is the size of the computation result, $\mathcal{C}_i(t)$ is the computation resources (CPU cycles/s) required by per bit of task (i.e., the computation intensity) , $\mathcal{C}_i^{\text{\text{\text{req}}}}(t)=\mathcal{D}_{i}^{\text{in}}(t)\cdot \mathcal{C}_i(t)$ is the amount of computation resources that is required to fulfill the task, and $T_{i}^{\max}$ is the maximum tolerable delay of the task.
}

\par\textit{\textbf{Server Model.}} First, the local of each edge server $j \in \{\mathcal{E}\}$ is denoted as $\mathbf{P}_j=(X_j, Y_j, 0)$. {\color{color} Furthermore, the heterogeneous computational capabilities of edge and cloud servers are characterized by different numbers of CPU cores and different amounts of computation resources of each CPU core, i.e., $j\in \{\mathcal{E}, o\}$ is characterized by $ \left (f_{j}^{\max},N_j^{\text{core}}\right)$, where $N_j^{\text{core}}$ denotes the number of CPU cores of server $j$ and $f_{j}^{\max}$ (cycles/s) denotes the maximal computational resource of server $j$. We consider that the servers are equipped with multi-core CPUs \cite{liu2019dynamic,Dai2019a} so that multiple tasks can be processed by a server in parallel, and each CPU core is dedicated to at most one task in each time slot \cite{liu2019dynamic}. }

\par {\color{color} \textit{\textbf{Strategy Variables.}} The following strategies are determined jointly. For task $\mathcal{T}_i(t)$, the offloading strategy is defined as a binary variant $s_i^a(t)\in\{0,1\}, a \in \{0,\mathcal{E},o\}$, where $a$ denotes the offloading destination of the task. Specifically, the task $\mathcal{T}_i(t)$ could be executed on vehicle $i$ ($s_{i}^{0}(t)=1$), on VEC server $j\in \mathcal{E}$ ($s_{i}^{j}(t)=1$), or on cloud server $o$ ($s_{i}^{o}(t)=1$). For server $j\in \{\mathcal{E},o\}$, the resource allocation strategy is denoted as $\{\left(f_{j,i}(t), c_{j, i}(t)\right)\}$, where $f_{j,i}(t)$ is the amount of computation resources provided by server $j$ to vehicle $i$ at time $t$, and $c_{j,i}(t)$ is the unit price of computation resources charged to vehicle $i$ for executing task $\mathcal{T}_i(t)$. }

%
%
\subsection{Communication Model}
\label{sec_communicationModel}

{ \color{color1}
\par This work mainly focuses on the task uploading from vehicles to edge servers, and the task downloading from edge servers to vehicles is omitted since the size of the computation outcome is generally much smaller than that of the computation input for most mobile applications \cite{guo2018collaborative,yi2019multi,Xu2019}. For task uploading, the power-domain NOMA is employed to support multiple requesting vehicles to simultaneously transmit their tasks to the RSU with which they are currently attached, considering the following reasons. In terms of the advantage of the power-domain NOMA, it has been recognized as a promising technique to achieve the capacity region for multi-user uplink communications in MEC networks \cite{Fang2021,Pan2019,Qian2021}, vehicular networks \cite{patel2021performance}, unmanned aerial vehicle (UAV) networks \cite{Liu2019a}, etc. Specifically, in the power-domain NOMA, multiple vehicles are able to share the same bandwidth resources, and they are distinguished in the power domain with the aid of the key technique of successive interference cancellation (SIC) \cite{Liu2022,liu2022developing}. In terms of the practicality of task uploading, when performing SIC, the SIC receiver decodes the stronger signals sequentially from the superimposed signals by treating the weaker signals as noise, implying that the performance of the power-domain NOMA is directly influenced by the capability of the SIC receivers \cite{Xu2020}.
Specifically, in the power-domain NOMA system, the SIC implementation has a linear computational complexity in the number of users \cite{luo2021error}. For the downlink scenario, the linear complexity may be a bottleneck since it requires the implementation of a sophisticated SIC scheme at each receiver with limited processing capabilities. However, for the uplink scenario that is mainly considered in this work, it is relatively more convenient and affordable for the RSU equipped with more powerful VEC servers when the capability of SIC receivers is taken into careful consideration.

\par In the considered NOMA-based VEC network, the details of employing the SIC for multi-vehicle task uploading is illustrated as follows. First, to capture the capability of the SIC receiver at the RSU side, we consider that each RSU $j$ is equipped with an $S_j$-SIC receiver \cite{Xu2020,Fang2020,Makki2020}, where $S_j$ means that the SIC receiver at RSU $j$ is capable of successively decoding the signals that are simultaneously transmitted by at most $S_j$ vehicles within its service range \footnote{\textcolor{color1}{The value of $S_j$ depends on the architecture types of SIC receivers (e.g., hard-based or soft-based SIC) in practical systems, which further relies on the capability of the RSU \cite{3gpp2018technical}.}}}. Furthermore, by using SIC, the channel gains of the uploading vehicles $\mathcal{V}_{j}(t)=\{1,\ldots,V_j(t)\}$ ($V_j(t)\leq S_j$) within the range of VEC server $j$ are first ordered as $g_{1,j}(t)\geq \cdots g_{i,j}(t) \cdots\geq g_{V_j,j}(t), \ \forall i \in \mathcal{V}_j(t)$ \cite{liu2018distributed}. To decode the signal of vehicle $i$, VEC server $j$ first decodes the stronger signals of vehicles $I^{\prime}<i$, then subtracts them from the superposed signal, and treats the weaker signals of vehicles $V_j(t)\geq I \geq i+1$ as interference. Therefore, the data transmission rate from vehicle $i$ to VEC server $j$ using NOMA can be calculated as:

\begin{equation}
\label{eq_dataRateUp}
	r_{i,j}(t)=B_{i,j}\cdot \log_2\left(1+\frac{P_{i}(t)\cdot g_{i,j}(t)}{N_0+\sum_{I=i+1}^{V_j(t)}P_{I}(t)\cdot g_{I,j}(t)}\right),
\end{equation}

\noindent where $B_{i,j}$ is the bandwidth, $P_{i}(t)$ is the transmit power of vehicle $i$, $g_{i,j}(t)$ is the channel power gain between vehicle $i$ and VEC server $j$, $N_0$ is the noise power, $P_{I}(t)$ is the transmit power of interference vehicle $I\in\mathcal{V}_j(t)$, and $g_{I,j}(t)$ is the channel power gain between the interference vehicle and the VEC server.

\par The channel power gain is calculated by incorporating the probabilistic LoS and NLoS transmissions into  both small-scale and large-scale fading \cite{yang2018dense}. For uplink communication, the channel power gain between vehicle $i$ and VEC server $j$ at time $t$ can be given as:
\begin{equation}
\label{eq_channelPowerGain}
	g_{i,j}(t)=p_{i,j}^{L}\cdot g_{i,j}^{L}(t) +(1-p_{i,j}^{L})\cdot  g_{i,j}^{NL}(t),
\end{equation}

\noindent where $p^{L}$ denotes the probability of LoS transmission and $g_{i,j}^{L}(t)/g_{i,j}^{NL}(t)$ denotes the channel power gain between vehicle $i$ and VEC server $j$ for LoS$/$NLoS transmission, which is calculated as:
{\color{color}
\setlength{\abovedisplayskip}{3pt}
\setlength{\belowdisplayskip}{3pt}
\begin{subnumcases}{\label{eq_channelPowerGain1}}
	$$g_{i,j}^{L}(t)=|h_{i,j}^{L}(t)|^2\cdot \left(L_{i,j}^{L}(t)\right)^{-1}\cdot 10^{\frac{-\chi_{\sigma}^{L}}{10}}$$,
	\label{eq_channelPowerGainLoS1}\\
	$$g_{i,j}^{NL}(t)=|h_{i,j}^{NL}(t)|^2\cdot \left(L_{i,j}^{NL}(t)\right)^{-1}\cdot 10^{\frac{-\chi_{\sigma}^{NL}}{10}}$$, \label{eq_channelPowerGainNLoS1}
\end{subnumcases}

\noindent where $h_{i,j}^{L}(t)/h_{i,j}^{NL}(t)$, $L_{i,j}^{L}(t)/L_{i,j}^{nL}(t)$, and $\chi_{\sigma}^{L}/\chi_{\sigma}^{NL}$ denote the components of \textit{small-scale fading}, \textit{pathloss}, and \textit{shadowing}, respectively for LoS/NLoS communication. For LoS communication, these components are detailed as follows. \textbf{i)} \textit{The small-scale fading} characteristic of the channel is captured by using a parametric-scalable and good-fitting generalized fading model, i.e., Nakagami-$m$ fading \cite{Boumaalif2022}.  Specifically, $h_{i,j}^{L}(t)/h_{i,j}^{NL}(t)$ follows a Nakagami distribution with fading parameter $m^{L}/m^{NL}$, i.e., $\begin{cases}h_{i,j}^{L}(t) \sim f_A\left(h_{i,j}^{L}(t),m^{L}\right)\\
h_{i,j}^{NL}(t) \sim f_A\left(h_{i,j}^{NL}(t),m^{NL}\right)\end{cases}$, $f_A(h,m)=\frac{2{\left(m\right)}^{m}\cdot h^{2m-1}\cdot e^{\left(-\frac{m\cdot  h^2}{\overline{p}}\right)}}{\Gamma\left(m\right) \cdot \overline{p}^m}$, where $0.5\leq m\leq 5$ denotes the Nakagami fading parameter, $\overline{p}$ is the average received power in the fading envelope, and $\Gamma(m)$ is the Gamma function. \textbf{ii)} \textit{The path loss} between vehicle $i$ and VEC server $j$ for LoS/NLoS communication can be calculated as $\begin{cases}L_{i,j}^{L}(t)=\frac{\left(4\pi \cdot d_0 \cdot f_c\right)^2}{c^2}\cdot \left(\frac{d_{i,j}(t)}{d_0}\right)^{\beta^{L}}\\L_{i,j}^{NL}(t)=\frac{\left(4\pi \cdot d_0 \cdot f_c\right)^2}{c^2}\cdot \left(\frac{d_{i,j}(t)}{d_0}\right)^{\beta^{NL}}\end{cases}$, where $f_c$ is the carrier frequency, $c$ is the speed of light, $d_0$ is the reference distance,  $d_{i,j}(t)$ is the distance between $i$ and $j$, and $\beta^{L}/\beta^{NL}$ is the path loss exponent for LoS/NLoS communication. \textbf{iii)} \textit{The shadowing} captures the signal attenuation caused by shadowing in transmission, which is a zero-mean Gaussian distributed random variable, i.e., $\begin{cases} \chi_\sigma^{L}(t)\sim\mathcal{N}\left(0,\left(\sigma^{L}\right)^2\right)\\
 \chi_\sigma^{NL}(t)\sim\mathcal{N}\left(0,\left(\sigma^{NL}\right)^2\right)
 \end{cases}$, where variance $\sigma^{L}/\sigma^{NL}$ denotes the standard deviation of shadowing for LoS/NLoS transmission \cite{yang2018dense}. Similarly, the channel power gain $g_{I,i}(t)$ between interference vehicle $I$ and VEC server $j$ can be easily obtained by replacing the corresponding parameters.
}

\subsection{Offloading Delay and Energy Consumption}
\label{sec_DelayModel}

\subsubsection{Offloading Delay}
\par For task $\mathcal{T}_i(t)$ generated by vehicle $i$ at time slot $t$, the service delay of completing the task depends on the offloading strategy $s_i^a$.

\par\textbf{\textit{Local Offloading.}} When task $\mathcal{T}_i(t)$ is processed by vehicle $i$ locally, the delay is given as:
\begin{equation}
\label{eq_time_local}
	T_{i}^{0}(t)=\frac{\mathcal{C}_{i}^{\text{\text{\text{req}}}}(t)}{f_{i(t)}},
\end{equation}
 
\noindent where $f_i(t)$ is the available computation capacity of vehicle $i$ at time $t$.


{
\par \textbf{\textit{Edge Offloading.}} When task $\mathcal{T}_i(t)$ is processed by VEC server $j$, the offloading delay mainly consists of transmission delay, computation delay and horizontal migration delay \footnote{{\color{color1}The delay of result feedback is omitted for horizontal migration as mentioned in Section \ref{sec_communicationModel}.}}, i.e.,
\begin{sequation}
	\label{eq_edge_delay}
	\begin{aligned}
	T_{i}^{j}(t)=\underbrace{T_{i,j^{\text{cur}}}^{\text{tran}}(t)}_{\text{\tiny{Transmission}}}+\underbrace{T_{i,j}^{\text{comp}}(t)}_{\text{\tiny{Computation}}}+\overbrace{\underbrace{T_{j^{\text{cur}},j}^{\text{hand}}(t) \mathbb{I}_{\left(j\neq j^{\text{cur}}\right)}}_{\text{\tiny{Task handover}}}+\underbrace{T_{j,j^{\text{arr}}}^{\text{hand}}(t) \mathbb{I}_{\left(j\neq j^{\text{arr}}\right)}}_{\text{\tiny{Result handover}}}}^{\text{{Horizontal migration} }}.
	\end{aligned}
\end{sequation}
}

\noindent \textbf{i)} $T_{i,j^{\text{cur}}}^{\text{tran}}(t)=\frac{\mathcal{D}_{i}^{\text{in}}(t)}{r_{i,j}(t)}$ is the delay for vehicle $i$ to transmit the task to VEC $j^{\text{cur}}$ where the vehicle is within the service area. \textbf{ii)} $T_{i,j}^{\text{comp}}(t)=\frac{\mathcal{C}_i^{\text{\text{\text{req}}}}(t)}{f_{j,i}(t)}$ is the execution delay of the task, where $f_{j,i}(t)$ is the computation resources that is allocated by VEC server $j$ to the task. {\color{color} \textbf{iii)} Similarly to \cite{Malik2021}, the horizontal migration delay incorporates task handover delay and result handover delay, i.e.,  $T_{j^{\text{cur}},j}^{\text{hand}}(t)=\frac{2\mathcal{D}_{i}^{\text{in}}(t)}{r_{f}}$ and $T_{j,j^{\text{arr}}}^{\text{hand}}(t)=\frac{2\mathcal{D}_{i}^{\text{out}}(t)}{r_{f}}$, where $r_{f}$ is the data rate of fiber link. On the one hand, when task is offloaded to the VEC server in whose coverage the vehicle is currently located, the task handover is used to forward the task firstly from the source VEC server $j^{\text{cur}}$ to the controller, and then from the controller to the selected VEC server $j$. On the other hand, when the result is not generated by the VEC server where the vehicle will arrive, the result handover is used to forward the result from VEC server $j$ to the controller, and from the controller to the destination VEC server $j^{\text{arr}}$ with which the vehicle will be attached. Note that $j^{\text{arr}}$ can be estimated by Eq. \eqref{eq_arrivalVEC}.} Therefore, the total delay of edge offloading can be obtained as:
{\color{color}
\begin{sequation}
	\label{eq_edge_delay1}
	T_{i}^{j}(t)={\frac{\mathcal{D}_{i}^{\text{in}}(t)}{r_{i,j}(t)}}+{\frac{2\mathcal{D}_{i}^{\text{in}}(t)}{r_{f}}\cdot \mathbb{I}_{\left(j\neq j^{\text{cur}}\right)}}+{\frac{\mathcal{C}_i^{\text{\text{\text{req}}}}(t)}{f_{j,i}(t)}}+{\frac{2\mathcal{D}_{i}^{\text{out}}(t)}{r_{f}}\cdot \mathbb{I}_{\left(j\neq j^{\text{arr}}\right)}}.
\end{sequation}
}

%
%
{\color{color}
\par \textbf{\textit{Cloud Offloading.}} When task $\mathcal{T}_i(t)$ is processed by cloud server $o$, the offloading delay mainly includes transmission delay, computation delay and vertical migration delay, i.e., 
\begin{sequation}
	\label{eq_cloud_delay}
	\begin{aligned}
		T_{i}^{o}(t)=\underbrace{T_{i,j^{\text{cur}}}^{\text{tran}}(t)}_{\text{{Transmission}}}+\underbrace{T_{i,o}^{\text{comp}}(t)}_{\text{{Computation}}}+\overbrace{\underbrace{T_{j^{\text{cur}},o}^{\text{hand}}(t) }_{\text{{Task handover}}}+\underbrace{T_{o,j^{\text{arr}}}^{\text{hand}}(t) }_{\text{{Result handover}}}}^{\text{{Vertical migration} }}.
	\end{aligned}
\end{sequation}

\noindent \textbf{i)} $T_{i,j^{\text{cur}}}^{\text{tran}}(t)=\frac{\mathcal{D}_{i}^{\text{in}}(t)}{r_{i,j}(t)}$ is the delay for vehicle $i$ to transmit the task to VEC server $j^{\text{cur}}$ within the current service area. \textbf{ii)} $T_{i,o}^{\text{comp}}(t)=\frac{\mathcal{C}_i^{\text{\text{\text{req}}}}(t)}{f_{o,i}(t)}$ is the computing delay of the task, wherein $f_{o,i}(t)$ is the computation resources that is allocated to the task. \textbf{iii)} similarly to edge offloading, the vertical migration delay consists of task handover delay and result handover delay, i.e.,
$T_{j^{\text{cur}},o}^{\text{hand}}(t)= \frac{\mathcal{D}_{i}^{\text{in}}(t)}{r_{c}}$ and $T_{o,j^\text{arr}}^{\text{hand}}(t)=\frac{\mathcal{D}_{i}^{\text{out}}(t)}{r_{c}}$, where $r_c$ is the data rate between the VEC server and the cloud \cite{Qian2019}. Therefore, the total delay of cloud offloading can be obtained as:
}
\begin{equation}
	\label{eq_cloud_delay1}
	T_{i}^{o}(t)=\frac{\mathcal{D}_{i}^{\text{in}}(t)}{r_{i,j}(t)}+\frac{\mathcal{C}_i^{\text{\text{\text{req}}}}(t)}{f_{o,i}(t)}+\frac{\mathcal{D}_{i}^{\text{in}}(t)+\mathcal{D}_{i}^{\text{out}}(t)}{r_{c}}.
\end{equation}

%
%

\subsubsection{Energy Consumption}
\label{sec_energyModel}

\par Completing task $\mathcal{T}_i(t)$ could incur additional costs for vehicle $i$, VEC server $j$ or cloud server $o$.

\par \textbf{\textit{Local Offloading.}} The energy consumption of vehicle $j$ to execute task $\mathcal{T}_i(t)$ locally is widely formulated as:
\begin{equation}
	\label{eq_energyLocalExe}
	E_{i}^{0}(t)=\alpha_i\cdot(f_{i}(t))^{\tau-1}\cdot \mathcal{C}_{i}^{\text{\text{\text{req}}}}(t),
\end{equation}
	
{\color{color} \noindent where $\alpha_i\geq0$ is the effective switched capacitance of vehicle $i$'s CPU that depends on the CPU chip architecture \cite{pan2021cost}, and $\tau>0$ is the constant that is typically set as $2$ or $3$ \cite{Zhang2018}.}


\par {\color{color}\textbf{\textit{Edge/Cloud Offloading.}} The energy consumption of server $j\in \{\mathcal{E},o\}$ to execute task $\mathcal{T}_i(t)$ can be given as:}
\begin{equation}
\label{eq_energyMecComp}
	E_{j}^{i}(t)=\alpha_j\cdot(f_{j,i}(t))^{\tau-1}\cdot \mathcal{C}_{i}^{\text{req}}(t),
\end{equation}

\noindent where $\alpha_j\geq0$ denotes the effective switched capacitance of server $j$'s CPU.

%
%

\subsection{Utility Model}
\label{sec_utilityModel}

{\color{color}\par When task $\mathcal{T}_i(t)$ is completed by using the offloading strategy $s_i^a(t)$ at time $t$, the utilities of vehicle $i$, VEC server $j$, and cloud server $o$ are formulated as follows.}

\subsubsection{Vehicle Utility}
\label{sec_vehicleUtility}

{\color{color} \par The utility obtained by vehicle $i\in \mathcal{V}$ from offloading task $\mathcal{T}_i(t)$ is formulated as:
\begin{equation}
	\label{eq_utilityVehicle}
	U_i^a(t)=w_i\cdot \Psi_i^a(t)-(1-w_i)C_i^a(t),
\end{equation}

\noindent where $\Psi_i^a(t)$ and $C_i^a(t)$ are the normalized satisfaction level of task completion delay and the normalized cost of task execution, respectively, when selecting offloading strategy $a$}, and $w_i$ is the weight coefficient of the satisfaction level. 

\par \textit{First}, the satisfaction function is a term that is widely used in economics, which is formulated as a logarithmic function that is convex and starts from zero. It has been employed to quantify the satisfaction level of task offloading \cite{Zhang2020a,zhao2019computation}. {\color{color} Based on these studies, the normalized satisfaction level can be calculated as:
\begin{equation}
	\label{eq_SatisficationVehicleNormal}
	\begin{aligned} 
		\Psi_i^a(t)=\frac{\log\left(1+\left(T_i^{\max}-T_i^a(t)\right)\right)}{\log\left(1+T_i^{\max}\right)},
	\end{aligned}
\end{equation}

\noindent where $T_i^a(t)$ is the total delay for completing task $\mathcal{T}_i(t)$, which can be obtained based on Eqs. \eqref{eq_time_local}, \eqref{eq_edge_delay1} and \eqref{eq_cloud_delay1}. 

\par \textit{Second}, the normalized cost of vehicle $i$ can be computed based on the energy consumption for local offloading or the payment for remote offloading, i.e.,
\setlength{\abovedisplayskip}{0pt}
\begin{subnumcases}{\label{eq_cost}C_i^a(t)=}
	$$\frac{E_i^{0}(t)}{E_i^{\max}}$$, \qquad  \qquad a=0,
	\label{eq_costA}\\
	$$\frac{c_{j,i}(t)\cdot f_{j,i}(t)}{C_i^{\max}}$$, \quad \ a=j, \ j\in\{\mathcal{E},o\} \label{eq_costB},
\end{subnumcases}	

\noindent where $E_i^0(t)$ (given by Eq. \eqref{eq_energyLocalExe}) is the computation energy consumption of vehicle $i$ and $E_i^{\max}$ is the energy constraint of vehicle $i$ \cite{ning2020intelligent}, $c_{j,i}(t)$ is the unit price of the computation resources charged by the VEC server or cloud server, and $f_{j,i}(t)$ is the computation resources allocated to vehicle $i$ by server $j$, and $C_i^{\max}$ is the budget of vehicle $i$ for the costs payed to the servers.

\par Therefore, $U_i^a(t)$ can be obtained based on Eqs.  \eqref{eq_energyLocalExe}, \eqref{eq_utilityVehicle} and \eqref{eq_cost} as:
\begin{sequation}
\label{eq_utilityVehicle1}
	\begin{aligned}
		&U_i^a(t)=w_i\cdot \frac{\log\left(1+\left(T_i^{\max}-T_i^a(t)\right)\right)}{\log\left(1+T_i^{\max}\right)}-(1-w_i)\\&\left(\frac{\alpha_i\cdot(f_{i}(t))^{\tau-1}\cdot \mathcal{C}_{i}^{\text{\text{\text{req}}}}(t)}{E_i^{\max}}\cdot \mathbb{I}_{\left(a=0\right)}+\frac{c_{j,i}(t)\cdot f_{j,i}(t)}{C_i^{\max}}\cdot \mathbb{I}_{\left(a=j\right)}\right).
	\end{aligned}
\end{sequation}

}

\subsubsection{Server Utility}
\label{sec_utility_server}

\par The utility of \textcolor{color}{server} $j\in \{\mathcal{E},o\}$ obtained by executing task $\mathcal{T}_i(t)$ is formulated as the revenue of task processing subtracting the cost of energy consumption:
\begin{equation}
	\label{eq_utilityMEC}
	\begin{aligned}
		U_j^i(t)&=w_j \cdot \mathcal{R}_j^i(t)-(1-w_j)C_j^i(t),
	\end{aligned}
\end{equation}

\noindent where $\mathcal{R}_j^i(t) $ and $C_j^i(t)$ are the normalized revenue and normalized cost of server $j$, respectively, and $w_j$ is the weight coefficient of revenue. 

\par {\color{color} \textit{First}, the normalized revenue of server $j$ to process task $\mathcal{T}_i(t)$ can be given as:
\begin{equation}
	\label{eq_revenueMEC}
	\begin{aligned}
		\mathcal{R}_j^i(t) = \frac{ c_{j,i}(t) \cdot f_{j,i}(t)}{C_j^{\max} \cdot f_{j}^{\max}},
	\end{aligned}
\end{equation}

\noindent where $ C_j^{\max} $ is the maximum unit price of server $j$'s computation resources. 

\par \textit{Second}, the normalized energy consumption of server $j$ to process task $\mathcal{T}_i(t)$ can be given as:
\begin{equation}
	\label{eq_costVEC}
	\begin{aligned}
		C_j^i(t) = \frac{E_j^i(t)}{E_j^{\max}},
	\end{aligned}
\end{equation}
 
\noindent where $E_j^i(t)$ (given by Eq. \eqref{eq_energyMecComp}) is the computation energy consumption of server $j$ and $E_j^{\max}$ is the energy constraint of server $j$.

\par Therefore, $U_j^i(t)$ can be obtained based on Eqs. \eqref{eq_utilityMEC}, \eqref{eq_revenueMEC}, \eqref{eq_costVEC}, and \eqref{eq_energyMecComp} as: 
\begin{sequation}
	\label{eq_utilityMEC1}
	\begin{aligned}
		U_j^i(t)&=w_j \cdot \frac{ c_{j,i}(t) \cdot f_{j,i}(t)}{ C_j^{\max} \cdot f_{j}^{\max}}-(1-w_j)\cdot\frac{\alpha_j\cdot(f_{j,i}(t))^{\tau-1}\cdot \mathcal{C}_{i}^{\text{req}}(t)}{E_j^{\max}}.
	\end{aligned}
\end{sequation}

\subsubsection{Social Welfare}
\label{sec_social_welfare}
\par Social welfare is employed in this work to quantify the system performance of computation resource allocation and task offloading for the servers and vehicles in the VEC network. Therefore, the social welfare at time $t$ can be given as follows:
\begin{sequation}
	\label{eq_socialwelfare}
	\begin{aligned}
			SW(t)&=\underbrace{\sum_{i\in \mathcal{V}}\sum_{a\in \mathcal{A}}p_i^{\text{gen}}(t) s_i^a(t) U_i^a(t)}_{\text{The total utility of vehicles}}+\underbrace{\sum_{j\in \{\mathcal{E},o\}}\sum_{i\in \mathcal{V}}s_i^j(t) p_i^{\text{gen}}(t) U_j^i(t)}_{\text{The total utility of servers}}
			\\&=\sum_{i\in \mathcal{V}}\sum_{a\in \mathcal{A}}p_i^{\text{gen}}(t)\cdot s_i^a(t)\cdot \left(U_i^a(t)+U_a^i(t)\right).
		\end{aligned}
\end{sequation}

\noindent Note that Eq. \eqref{eq_socialwelfare} incorporates the local, \textcolor{color}{edge, and cloud} offloading strategies. If vehicle $i$ offloads the task locally at time $t$, i.e., $a=0$, then the utility of server $j$ to provide service for vehicle $i$ is 0, i.e., $U_{0}^{i}(t)=0$.
}

%
%
\section{Problem Formulation}
\label{sec_problemFormulation}

\par The objective of this work can be transformed to maximize the social welfare over $\mathbf{T}$ slots by jointly optimizing the task offloading strategy \textcolor{color}{$S_t = \{s_i^a(t)\}_{i\in \mathcal{V}, a \in \mathcal{A}, t\in \mathbf{T}}$}, and the computation resource allocation and pricing strategy \textcolor{color}{$S_c= \{f_j^i(t),c_j^i(t)\}_{j\in \{\mathcal{E},o\}, i\in \mathcal{V}, t\in \mathbf{T}}$. Therefore, JRATOP can be formulated as follows:}
\begin{sequation}
\label{eq_problem}
	\begin{aligned}
		\textbf{P}: \quad &\max_{S_t, S_c}  \ \sum_{t=t_0}^\mathbf{T}SW(t)\\
		\text{s.t.}\ \ 
		&\text{C1}: s_i^a(t)\in\{0,1\}, \ \forall i\in \mathcal{V}	a\in \mathcal{A}\\
		&\text{C2}: \sum_{a\in \mathcal{A}}s_i^a(t)\leq 1, \ \forall  i\in \mathcal{V},a\in \mathcal{A}\\
		&\textcolor{color}{\text{C3}: p_i^{\text{gen}}(t)=\{0,1\}, \ \forall i\in \mathcal{V}}\\
		&\textcolor{color}{\text{C4}: s_i^a(t)\cdot T_i^a(t)\leq T_i^{\max}(t), \ \forall i\in \mathcal{V}, \forall j\in \{\mathcal{E}, o\}, a\in \mathcal{A}}\\
		&\textcolor{color}{\text{C5}: s_i^a(t) \cdot T_{i,j^{\text{cur}}}^{\text{tran}}(t) \leq T_{i,j^{\text{cur}}}^{\text{soj}}, \ \forall i\in \mathcal{V}, \forall j^{\text{cur}}, a \in \{\mathcal{E}, o\}},\\
		& \textcolor{color}{\text{C6}: s_i^a\cdot T_{j,j^{\text{arr}}}^{\text{hand}}(t) \leq T_{i,j^{\text{arr}}}^{\text{soj}}, \ 
			\forall i\in \mathcal{V}, \forall j\in \{\mathcal{E}, o\}, \forall a\in\{\mathcal{E},o\}}\\
		&\textcolor{color}{\text{C7}: v^{\min}\leq v_i(t) \ \leq v^{\max}, \ \forall i\in \mathcal{V}}\\
		&\textcolor{color}{\text{C8}: \sum_{i\in\mathcal{V}} s_i^{j}(t)\cdot f_{j,i}(t) \leq  f_j^{\max}, \  \forall j\in \{\mathcal{E},o\}}\\
		&\textcolor{color}{\text{C9}: \sum_{i \in \mathcal{V}} s_i^{j}(t) \leq N_{j}^{\text{core}}, \ \forall j \in \{\mathcal{E},o\}}\\
		&\textcolor{color}{\text{C10}: s_i^0(t)\cdot E_i^0(t)\leq 	E_i^{\max}, \ \forall i\in \mathcal{V}}\\
		&\textcolor{color}{\text{C11}: \sum_{i\in\mathcal{V}} s_i^j(t)\cdot E_j^i(t)\leq 	E_j^{\max}, \ \forall j\in \{\mathcal{E},o\}}\\
		&\textcolor{color}{\text{C12}: s_i^a(t)\cdot c_{a,i}(t)\cdot f_{a,i}(t)\leq C_i^{\max}, \ \forall i\in \mathcal{V}, \forall a\in{\{\mathcal{E},o\}}} \\
	\end{aligned}
\end{sequation}

\noindent \textcolor{color}{where $T_{i,j^{\text{cur}}}^{\text{soj}}$ and $T_{i,j^{\text{arr}}}^{\text{soj}}$ can be obtained by Eq. \eqref{eq_sojourn_delay}, and $j^{\text{arr}}$ can be obtained by Eq. \eqref{eq_arrivalVEC}.} Constraints C1 and C2 are the values of offloading strategies of vehicles, which indicates that the vehicle can only select one strategy as its offloading decision. {\color{color} Constraint C3 represents each vehicle generates at mos time slot. Constraint C4 is the delay constraint of the task, which guarantees that the task is completed before the deadline. Constraint C5 ensures that the task uploading is completed before vehicle $i$ moving out of the coverage of the connected VEC server $j^{\text{cur}}$. Constraint C6 guarantees that the result dispatch is completed before vehicle $i$ moving out of the coverage of VEC server $j^{\text{arr}}$ with which it will be attached. VEC server $j^{\text{arr}}$ can be obtained by substituting  $T_{i}^{\text{move}}=T_{i,{{j}^{\text{cur}}}}^{\text{tran}}(t)+T_{{{j}^{\text{cur}}},j}^{\text{hand}}(t){{\mathbb{I}}_{\left( j\ne {{j}^{\text{cur}}} \right)}}+T_{i,j}^{\text{comp}}(t)$ into Eq. \eqref{eq_arrivalVEC}, where $T_{i}^{\text{move}} (j\in \{\mathcal{E},o\})$ is the moving duration before the result of task is generated.} Constraint C7 poses constraints on the velocity of each vehicle. Constraints C8 and C9 limit the computation resources and the number of CPU cores, respectively, for each server. {\color{color} Constraints C10 and C11 constrain the energy budgets of vehicles and servers, respectively. Constraint C12 represents a vehicle's maximum payment for the computational resources provided by servers.}

\begin{theorem}
	\label{lemma_NP}
	The problem $\textbf{P}$ formulated in Eq. \eqref{eq_problem} is NP-hard. 
\end{theorem}

\begin{proof}
	The detailed proof is given in Appendix A of the supplemental material.
\end{proof}

%
%
\section{BARGAIN-MATCH}
\label{sec_jointOffloading}

{\color{color}
\par To solve problem $\mathbf{P}$, the algorithm of BARGAIN-MATCH is proposed for resource allocation and task offloading by using the bargaining and matching schemes. BARGAIN-MATCH  mainly consists of the following two parts. i) For intra-server resource allocation, the bargaining game is used to stimulate the negotiation between the requesting vehicle and server on resource allocation and price in Section \ref{sec_bargain}. ii)  For inter-server task offloading, the many-to-one matching is constructed between tasks and servers to stimulate both edge-edge collaboration for horizontal task migration and edge-cloud collaboration for vertical task migration. 
}

%
%
\subsection{Resource Allocation and Pricing: A Bargaining Game-based Scheme}
\label{sec_bargain}

%

%
%
\subsubsection{Fundamentals}
\label{sub_fundamentalsBargaining}

\par {\color{color} Denote the set of requesting vehicles that have tasks to upload at time $t$ as $\mathcal{V}^{\text{req}}(t)=\{i|i\in\mathcal{V},p_i^{\text{gen}}(t)=1,s_i^0(t)=0\}$. A finite-time bargaining game is constructed to stimulate the negotiation between the requesting vehicle and its expected server on the decisions of resource allocation and pricing in period $\Delta t$. The bargaining game is defined as a triplet of $(\mathcal{P}_B, \mathcal{S}_B, \mathcal{U}_B)$, which is as follows:}

\begin{itemize}[]
	\item $\mathcal{P}_B=\{i\in \mathcal{V}^{\text{req}}(t), j \in \{\mathcal{E},o\}$ denotes the parties, i.e., a seller and a buyer. Vehicle $i$ acts as a buyer that aims to offload its task by buying the required computation resources from the server. Moreover, server $j$ acts as a seller who aims to execute the task of the vehicle by allocating the computation resources and charging from the vehicle.
	
	\item $\mathcal{S}_B=\{f_{j,i}(t),c_{j,i}(t)\}$ denotes the set of strategies. The strategy of vehicle $i$ is to request the satisfactory amount of computation resources from the server, and the strategy of server $j$ is to decide the satisfactory price of the computation resources that are sold to vehicle $i$.	
	
	\item $\mathcal{U}_B=\{U_i^j(t), U_j^i(t)\}$ denotes the utilities of vehicle $i$ and server $j$, wherein $U_i^j(t)$ and $U_j^i(t)$ are given in Eqs. \eqref{eq_utilityVehicle1} and \eqref{eq_utilityMEC1}, respectively.
\end{itemize}

%
%
\subsubsection{Resource Allocation and Pricing}
\label{sec_optAlloPrice}

\par In this section, the optimal resource allocation and pricing strategies are presented.

{\color{color}
\begin{theorem}
	\label{lemma_convex}
	For vehicle $i$, the expected optimal amount of computation resources it intends to request from target server $j$ to offload task $\mathcal{T}_i(t)$ is obtained as: $f_{j,i}^*(t)=\frac{2w_i\cdot C_i^{\max}}{F(var_1,c_{j,i}(t))-\log\left(1+T_i^{\max}(t)\right)\cdot c_{j,i}(t)\cdot\left(1-w_i\right)}$, where $F(var_1,c_{j,i}(t))$ is given by \textbf{Eq. (24)}, which is the function of variables $var_1$ and $c_{j,i}(t)$. $var_1={\frac{\mathcal{D}_{i}^{\text{in}}(t)}{r_{i,j}(t)}}+{\frac{2\mathcal{D}_{i}^{\text{in}}(t)}{r_{f}}\cdot \mathbb{I}_{\left(j\neq j^{\text{cur}}\right)}}+{\frac{2\mathcal{D}_{i}^{\text{out}}(t)}{r_{f}}\cdot \mathbb{I}_{\left(j\neq j^{\text{arr}}\right)}}$ for edge server $j\in \mathcal{E}$, and $var_1=\frac{\mathcal{D}_{i}^{\text{in}}(t)}{r_{i,j}(t)}+\frac{\mathcal{D}_{i}^{\text{in}}(t)+\mathcal{D}_{i}^{\text{out}}(t)}{r_{c}}$ for cloud server $j=o$.
\end{theorem}
}

\begin{proof}
	\label{proof_concave}
	The detailed proof is given in Appendix B of the supplemental material.
\end{proof}

\newcounter{mytempeqncnt}
\begin{figure*}[!h]
	\label{eq:F}
	\normalsize
	\setcounter{mytempeqncnt}{\value{equation}}
	\setcounter{equation}{23}
	{\color{color}
	\begin{sequation}
		F(var_1,c_{j,i}(t))=\sqrt{\frac{c_{j,i}(t) \log(1+T_i^{\max}(t))(1-w_i) \left(c_{j,i}(t) \cdot \mathcal{C}_i^{\text{req}}(t) \cdot \log(1+T_i^{\max}(t)) (1-w_i)+4C_i^{\max} w_i\left(1+T_i^{\max}(t)-var_1\right)\right)}{\mathcal{C}_i^{\text{req}}(t)}}
	\end{sequation}
    }
	\setcounter{equation}{\value{mytempeqncnt}}
	\hrulefill 
\end{figure*}
\setcounter{equation}{24}

{\color{color}
\begin{lemma}
	\label{lemma_priceRange}	
	 When the controller decides to offload task $\mathcal{T}_i(t)$ to server $j\in\{\mathcal{E},o\}$ at time slot $t$, there exist a lower bound and an upper bound for the unit price of the computation resources that are allocated to the task, i.e., $c_{j,i}^{\min}(t) \leq c_{j,i}(t)\leq c_{j,i}^{\max}(t)$, where		$c_{j,i}^{\min}(t)=\frac{ (1-w_j)\alpha_j\cdot(f_{j,i}(t))^{\tau-2}\cdot \mathcal{C}_{i}^{\text{req}}(t)\cdot C_j^{\max}\cdot f_j^{\max}}{w_j\cdot E_j^{\max}}$ and $	c_{j,i}^{\max}(t)= \frac{w_i\log\left(1+T_i^{\max}-T_i^j(t)\right)\cdot C_i^{\max}}{(1-w_i)\cdot f_{j,i}(t)\cdot \log\left(1+T_i^{\max}\right)}$.
\end{lemma}
}

\begin{proof}
	The detailed proof is given in Appendix C of the supplemental material.
\end{proof}

\par According to Lemma \ref{lemma_priceRange}, the bid-ask spread can be obtained as \textcolor{color}{$\Delta c_{j,i}(t) = c_{j,i}^{\max}(t)-c_{j,i}^{\min}(t) $.} The trade on the computation resources between vehicle $i$ and VEC server $j$ is modeled as the process of two players bargaining over the pie of size $\Delta c_{j,i}(t)$ according to \cite{rubinstein1982perfect}. Obviously, both bargainers desire to reach an agreement on the proposal of the partition earlier because their utilities will be discounted over time. Therefore, the discount factor is introduced in the bargaining game to describe the discount of the partition in the future. The discount factor captures the patience levels of the bargainers. In other words, the smaller discount factor indicates that the players are impatient with the delay of the negotiation. Accordingly, the discount factors of vehicle $i$ and server $j$ are formulated as follows:
\begin{equation}
\label{eq_discountVehicle}
	\epsilon_{i}(t)=1-\frac{T_{i,j}^{\text{tran}}(t)}{T_i^{\max}(t)},
\end{equation}
\begin{equation}
\label{eq_discountMEC}
	\epsilon_{j}(t)=1-\frac{T_{i,j}^{\text{comp}}(t)}{T_i^{\max}(t)}.
\end{equation}

\noindent For vehicle $i$, it is more impatient if it takes longer time to upload the task to the VEC server, leading to the lower value of $\epsilon_i(t)$. For VEC server $j$, it is more impatient if it takes longer time to execute the task. Besides, the larger task deadline $T_i^{\max}(t)$ indicates that the players have higher endurance to the delay, leading to higher values of $\epsilon_i(t)$ and $\epsilon_j(t)$.

\begin{lemma}
	\label{theorem_ne}
	\par The bargaining game has an unique perfect partition. In the period $T^{b}\in {T}(n)$ in which vehicle $i$ makes a proposal, the optimal partitions are given as:
	$$
		\begin{cases}
			{\delta_{i}^{i}}^{*}(t)=\epsilon_{i}(t)-\frac{\left(1-\epsilon_{i}(t)\right)\left(1-\left(\epsilon_{i}(t)\epsilon_{j}(t)\right)^{\lceil \frac{T^b}{2}\rceil}\right)}{1-\epsilon_{i}(t)\epsilon_{j}(t)},\\
			{\delta_{j}^{i}}^{*}(t)=\frac{\left(1-\epsilon_{i}(t)\right)\left(2-\epsilon_{i}(t)\epsilon_{j}(t)-\left(\epsilon_{i}(t)\epsilon_{j}(t)\right)^{\lceil \frac{T^b}{2}\rceil}\right)}{1-\epsilon_{i}(t)\epsilon_{j}(t)}.
		\end{cases}
	$$
	
	\noindent In the period $T^b\in T(n)$ when server $j$ makes a proposal, the optimal partitions are given as:
	
	$$
	 \begin{cases}
		{\delta_{i}^{j}}^{*}(t)=\frac{\left(1-\epsilon_{j}(t)\right)\left(1-\left(\epsilon_{i}(t) \epsilon_{j}(t)\right)^{\lceil{\frac{T^{b}}{2}\rceil}}\right)}{1-\epsilon_{i}(t) \epsilon_{j}(t)},\\
		{\delta_{j}^{j}}^{*}(t)=\frac{\epsilon_{j}(t)\left(1-\epsilon_{i}(t)\right)-\left(1-\epsilon_{j}(t)\right)\left(\epsilon_{i}(t) \epsilon_{j}(t)\right)^{\lceil{\frac{T^{b}}{2}\rceil}}}{1-\epsilon_{i}(t) \epsilon_{j}(t)}.
	\end{cases}
	$$
\end{lemma}

\begin{proof}
	\label{eq_optPartition}
	The detailed proof is given in Appendix D of the supplemental material.
\end{proof}

\begin{theorem}
	\label{lemma_pricing}	
	\par The optimal price of computation resource $c_{j,i}^*(t)$ that the VEC server $j$ charges vehicle $i$ is obtained as: 1) in the period in which vehicle $i$ makes a proposal, \textcolor{color}{ $c_{j,i}^*(t)=c_{j,i}^{\max}(t)-\Delta c_{j,i}(t)\cdot {\delta_{i}^{i}}^{*}(t)$; 2) in the period in which VEC server $j$ makes a proposal, $	c_{j,i}^*(t)=c_{j,i}^{\max}(t)-\Delta c_{j,i}(t) \cdot {\delta_{i}^{j}}^{*}(t)$.}
\end{theorem}

\begin{proof}
	\label{eq_optPrice}
		The detailed proof is given in Appendix E of the supplemental material.
\end{proof}

{\color{color}
\begin{corollary}
\label{cor_deal}
\par It can be concluded that a deal on the optimal computation resource allocation $f_{j,i}^*(t)$ and pricing $c_{j,i}^*(t)$ can be achieved by Eqs. \eqref{eq_optAllo} and \eqref{eq_optAlloPrice}, respectively, when $f_{j,i}^*(t)$ and $c_{j,i}^*(t)$ satisfy Theorems \ref{lemma_convex} and \ref{lemma_pricing} simultaneously.

\begin{sequation}
	\label{eq_optAllo}
	f_{j,i}^*(t)=\frac{2w_i\cdot C_i^{\max}}{F\left(var_1,c_{j,i}^*(t)\right)-\log\left(1+T_i^{\max}(t)\right)\cdot c_{j,i}(t)\cdot\left(1-w_i\right)},
\end{sequation}

\begin{subnumcases}{\label{eq_optAlloPrice}c_{j,i}^*(t)=}
	$$c_{j,i}^{\max}(t)-\Delta c_{j,i}(t)\cdot {\delta_{i}^{i}}^{*}(t)$$,  \label{eq_optAlloPriceA}\\
	$$c_{j,i}^{\max}(t)-\Delta c_{j,i}(t)\cdot  {\delta_{i}^{j}}^{*}(t)$$.\label{eq_optAlloPriceB}
\end{subnumcases}
\end{corollary}
}

\par Based on the result of optimal resource allocation and pricing, a negotiation approach between vehicle $i$ and VEC server $j$ that intend to conclude a transaction for task offloading is presented in Algorithm \ref{algo_allo_price}. The negotiation is mainly based on the pricing rule, which is defined as follows.

\begin{definition}
	\label{def_pricingRule}
	
	Pricing Rule. Before the transaction, \textcolor{color}{VEC server} $j$ initially sets the optimal resource allocation $f_{j,i}^*(t)$ as its available computation resource. Then, the resource price is updated based on the following rules.
	
	\begin{itemize}
		\item If $U_i^j(t)>0$ $\&\&$ $U_j^i(t)>0$, an agreement on the optimal resource allocation and pricing is reached.
		
		\item If $U_i^j(t)>0$ $\&\&$ $U_j^i(t)<0$,  vehicle $i$ proposes the optimal price based on Eq. \eqref{eq_optAlloPriceA}.
		
		\item If $U_i^j(t)<0$ $\&\&$ $U_j^i(t)>0$, VEC server $j$ proposes the optimal price based on Eq. \eqref{eq_optAlloPriceB}.
		
		\item If $U_i^j(t)<0$ $\&\&$ $U_j^i(t)<0$,  either VEC server $j$ or vehicle $i$ can propose the optimal price.
	\end{itemize}
\end{definition}

\begin{algorithm}[]	
	\label{algo_allo_price}	
	\SetAlgoLined
	\KwIn{Vehicle $i$, VEC server $j$}
	\KwOut{The optimal reource allocation and pricing strategy  $s_{j,i}^*$ between $i$ and $j$}
	\textbf{ Initialization:} 
	$U_i^j(t)= 0$; $U_j^i(t)= 0$\;
	\textcolor{color}{VEC server} $j$ sets the optimal allocation as \textcolor{color}{$f_{j,i}^*(t) = f_j^{avl}$}\;
	Calculate the optimal price $c_{j,i}^*(t)$ based on Eq. \eqref{eq_optAlloPrice}\;
	Calculate $U_i^j(t)$ for vehicle $i$ based on \textcolor{color}{Eq. \eqref{eq_utilityVehicle1}}\;
	Calculate $U_j^i(t)$ for VEC server $j$ based on \textcolor{color}{ Eq. \eqref{eq_utilityMEC1}}\;
	
	\uIf{$U_i^j(t)>0$ $\&\&$ $U_j^i(t)>0$}
	{
		\tcp{An agreement is reached. }
		\Return{ $s_{j,i}^*(t)=(f_{j,i}^*(t), c_{j,i}^*(t))$}\;
	}
	\uElseIf{$U_i^j(t)>0$ $\&\&$ $U_j^i(t)<0$}
	{
		Vehicle $i$ proposes the optimal price $c_{j,i}^*(t)$  based on Eq. \eqref{eq_optAlloPriceA}\;
	}
	\uElseIf{$U_i^j(t)<0$ $\&\&$ $U_j^i(t)>0$}
	{
		VEC server $j$ proposes the optimal price \textcolor{color}{$c_{j,i}^*(t)$  based on Eq. \eqref{eq_optAlloPriceB}}\;
	}
	\uElse{
		Either of the players can propose the optimal price \textcolor{color}{$c_{j,i}^*(t)$ based on  Eq. \eqref{eq_optAlloPrice}}\;
	}
	\Return{\textcolor{color}{$s_{j,i}^*=(f_{j,i}^*(t), c_{j,i}^*(t))$}}\;
	\caption{Resource Allocation and Pricing.}
\end{algorithm}	


%
%
\subsection{Offloading Strategy Selection: A Matching-based Scheme}
\label{sec_matching}

\par {\color{color} Denote the tasks that have not been decided where to offload as $\mathcal{T}^{\text{req}}(t)=\{\mathcal{T}_i(t)| i\in \mathcal{V}\}$. Then the offloading strategy of each task $k\in\mathcal{T}^{\text{req}}(t)$ is decided using a many-to-one matching scheme. Specifically, if task $k$ is matched with VEC server $j\in \mathcal{E}$, it will be offloaded horizontally from VEC server $j^{\text{cur}}$ to VEC server $j$; if it is matched with server $o$, it will be offloaded vertically from VEC server $j^{\text{cur}}$ to the cloud.
}

%
%

\subsubsection{Fundamentals}
\label{sub_fundamentals}

\par The matching is described as a triplet of $ (\mathcal{P}_M,\Omega, \Phi)$: 

\begin{itemize}
	\item $\mathcal{P}_M=\left(\mathcal{T}^{\text{req}}(t), \{\mathcal{E},o\}\right)$ denotes two disjoint sets of players where $\mathcal{T}^{\text{req}}(t)=\{\mathcal{T}_i(t)| i\in \mathcal{V}\}$ is the set of tasks {\color{color} that have not been decided where to offload currently}, and $\{\mathcal{E},o\}$ is the set of servers.
	
	\item\textcolor{color}{$\Omega=\left(\Omega_k^{j},\Omega_j^{k}\right)$ denotes the preference lists of the tasks and servers. Each task $k\in \mathcal{T}^{\text{req}}(t)$ has a descending ordered preferences on the servers, i.e., $\Omega_k^j=\{j|j\in\{\mathcal{E},o\}, j\succ_{i}{j^\prime}\}$, where $\succ_k$ denotes the preference of task $k$ towards the servers. Furthermore, each server $j \in \{\mathcal{E},o\}$ has a descending ordered preference list over the tasks, i.e., $\Omega_j^k=\{k\in \mathcal{T}^{\text{req}}(t), k \succ_{j} {k^\prime}$. }
	
	\item \textcolor{color}{$\Phi\subseteq \{k| k\in\mathcal{T}^{\text{req}}(t)\} \times \{\mathcal{E},o\}$ is the many-to-one matching between the tasks and servers. Each task $k\in \mathcal{T}^{\text{req}}(t)$ can be matched with at most one server, i.e., $\Phi(k)\in \{\mathcal{E},o\}$, and each server $j\in\{\mathcal{E},o\}$ can be matched with multiple tasks, i.e., $\Phi(j)\subseteq \{k| k\in\mathcal{T}^{\text{req}}(t)\} $}.
	
\end{itemize}



\subsubsection{Preference List Construction}
\label{sub_preConstruct}

\par The preference lists are constructed as follows.
{\color{color}
\begin{enumerate}[label=\alph*)]
	\item Predict the optimal resource allocation $f_{j,i}^*(t)$ and price $c_{j,i}^*(t)$ allocated by server $j$ to each task $k=\mathcal{T}_i(t)\in \mathcal{T}^{\text{req}}(t)$ based on Algorithm \ref{algo_allo_price}. Note that $k=\mathcal{T}_i(t)$ means that task $k$ is generated by vehicle $i$ at time $t$.
	
	\item Calculate the values of preference for each server $j\in \{\mathcal{E},o\}$ on tasks $k\in \mathcal{T}^{\text{req}}(t)$ as:
	\begin{equation}
		\label{eq_preValueM}
		\rho_{j}(k)=U_j^i(t), \ k = \mathcal{T}_i(t),
	\end{equation}
	\noindent where $U_j^i(t)$ is the utility of server $j$ to process task $k$ of vehicle $i$ (in Eq. \eqref{eq_utilityMEC}). 
	
	\item Construct the matching list for each server $j\in \{\mathcal{E},o\}$ by ranking the preference values as a descending order: 
	\begin{equation}
		\label{eq_preOrderM}
		\begin{aligned}
			\rho_{j}(k)>\rho_{j}(k^{\prime}) \Leftrightarrow k \succ_j k^{\prime}, \ \Omega_j^k=\{k,k^{\prime}\}
		\end{aligned}
	\end{equation}

	\item Calculate the preference values of each task $k\in \mathcal{T}^{\text{req}}(t)$ on each server $j\in \{\mathcal{E}, o\}$ as:
	\begin{equation}
		\label{eq_preValueV}
		\rho_{k}(j)=U_i^j(t), \ k=\mathcal{T}_{i}(t)
	\end{equation}
	
	\noindent where $U_i^j(t)$ is the utility obtained by vehicle $i$ when its task $k$ is processed by server $j$ at time $t$.

	\item Construct the matching list for each task $k\in \mathcal{T}^{\text{req}}(t)$ by ranking the preference values as a descending order:
	\begin{equation}
		\label{eq_preOrderV}
		\rho_{k}(j)>\rho_{k}(j^{\prime}) \Leftrightarrow j\succ_i j^{\prime}, \ \Omega_k^j=\{j,j^{\prime}\}.
	\end{equation}	
\end{enumerate}
}

\begin{algorithm}[]	
	\label{algo_matching}	
	\SetAlgoLined
	\KwIn{\textcolor{color}{The requesting vehicle set $\mathcal{V}^{req}(t)$, task set $\mathcal{T}^{\text{req}}(t)=\{\mathcal{T}_i(t)| i\in \mathcal{V}^{\text{req}}(t)$, and server set $\{\mathcal{E},o\}$}}
	\KwOut{\textcolor{color}{The optimal matching list $\Phi(t)$, the offloading strategy  ${S}^*_{\text{off}}(t)$, and the resource allocation strategy ${S}^*_{\text{all}}(t)$. }
	}
	\textbf{ Initialization:} 
	$\mathcal{T}^{\text{rej}}(t)=\mathcal{T}^{\text{req}}(t)$, $\mathcal{E}^{\prime}=\emptyset$, $\Phi^*=\emptyset$\;
	\tcp{Preference lists construction}
\For{\textcolor{color}{$k\in \mathcal{T}^{\text{req}}(t)$}}
	{
		\For {\textcolor{color}{$j\in \{\mathcal{E},o\}$}}
		{
			Call Algorithm \ref{algo_allo_price} for 	\textcolor{color}{$s_{j,i}^*(t)=\left(f_{j,i}^*(t), c_{j,i}^*(t)\right)$\;}
			Calculate the preference values of server $j$ on task $k$ as Eq. \eqref{eq_preValueM}\;
			Construct the matching list of server $j$ based on Eq. \eqref{eq_preOrderM}\;
			{\color{color}Calculate the preference values of task $k$ on server $j$ as Eq. \eqref{eq_preValueV}\;
			Construct the matching list of task $k$ as Eq. \eqref{eq_preOrderV}\;}
		}
	}
	
	\tcp{Matching construction}
	\While{\rm{There exists}}
	{
		\For{\rm {task} \color{color} {$k=\mathcal{T}_i(t)\in \mathcal{T}^{\text{rej}}(t)$}}
		{
			{\color{color}
				Select the most preferred server $j^{\prime}$\;
				Update the preference list of $k$ as Eq. \eqref{eq_addMatchVeh}\; 
				Update the preference list of $j^{\prime}$ as Eq. \eqref{eq_addMatch}\;
			}
		}
		\For{\rm {server} {\color{color}$j\in \{\mathcal{E}^{\prime},o\}$} \rm{that receives new requests}}
		{
			Update matching list based on Eq. \eqref{eq_decideMatch}\;
			{\color{color}Update $\mathcal{T}^{\text{rej}}(t)$ based on \eqref{eq_delete_veh1}\; }
			\If{$D_j\neq \emptyset$}
			{
				\For{\rm {task} $k\in D_j$}
				{
					Update the preference list as Eq. \eqref{eq_deletMecA}\;
					Update the matching list as Eq. \eqref{eq_deletMecB}\;
				}
			}
		}     
	}
	\Return \textcolor{color}{{$\Phi(t)$, ${S}^*_{\text{off}}(t)=\{s_i^{a^*}(t)| a^*=\Phi^*(k), k=\mathcal{T}_i(t)\in \mathcal{T}^{\text{req}}(t)\}$, ${S}^*_{\text{all}}(t)=\{s_{j,i}^*(t)|j\in \{\mathcal{E},o\}, k=\mathcal{T}_i(t)\in \Phi^*(j)\}$}}\;
	\caption{Matching algorithm for tasks and servers in time slot $t$:}
	\vspace{1pt}
\end{algorithm}


\subsubsection{Matching Construction}
\label{sub_matchConstruct}

\par The matching scheme is implemented as follows.
\begin{enumerate}[label=\alph*)]
	{\color{color}\item The rejected set is initialized as $\mathcal{T}^{\text{rej}}(t)=\mathcal{T}^{\text{req}}(t)$.
	
	\item Select the most preferred server $j^{\prime} = \Omega_k^j[0]$ for each task $k\in \mathcal{T}^{\text{rej}}$, and adds $j^{\prime} $ to the matching list of each task temporarily:	
	\begin{equation}
	\label{eq_addMatchVeh}
		\Phi(k)=\Phi(k)\cup j^{\prime}.
	\end{equation}

	\item If task $k$ prefers server $j^{\prime}\in \{\mathcal{E},o\}$, add task $k$ to the matching list of server $j^{\prime}$ temporarily:	
	\begin{equation}
	\label{eq_addMatch}
		\Phi(j^{\prime})=\Phi(j^{\prime} )\cup k.
	\end{equation}
	
	\item Update the matching list $\Phi(j)$ of each server $j\in\{\mathcal{E}, o\}$ by remaining the top-$n$ most preferred tasks and removing the less preferred tasks:		
	\begin{subnumcases}{\label{eq_decideMatch}}
		$$|\Phi(j)|\leq n\leq N_{j}^{\text{core}}, \notag \\ \sum_{k \in \Phi(j)} f_{j,i}^*(t)\leq f_j^{\max}, k=\mathcal{T}_{i}(t)$$, \label{eq_decideMatchA}\\
		$$\Phi(j)=\Phi(j)\ \backslash \ D_j$$\label{eq_decideMatchB},
	\end{subnumcases}

 \noindent where $n$ is the number of server $j$'s CPU cores that are idle at the current time, $N_{j}^{\text{core}}$ is the total number of server $j$'s CPU cores, and $D_j$ is the set of less preferred tasks of server $j$.
	
	\item Add the tasks in $D_j$ to the rejected set: 
	\begin{equation}
	\label{eq_delete_veh1}
		\mathcal{T}^{\text{rej}}(t)=\mathcal{T}^{\text{rej}}(t)\cup D_j.
	\end{equation}
	
	\item Update the preference list and matching list of each deleted task $k\in D_j$: 	
	\begin{subnumcases}{\label{eq_deletMec}}
		$$\Omega_k^j=\Omega_k^j\ \backslash \ \{j\}$$,  \label{eq_deletMecA}\\
		$$\Phi(k)=\Phi(k) \ \backslash \ \{j\}$$\label{eq_deletMecB},
	\end{subnumcases}

	\item For the tasks $k=\mathcal{T}_i(t)\in \mathcal{T}^{\text{rej}}(t)$ that are deleted in the last iteration, repeat the steps b) to d) until all tasks have been matched with a server, or the unmatched tasks have been rejected by all servers.}
\end{enumerate}


%
%

\subsection{Main Steps of BARGAIN-MATCH and Analysis}
\label{sec_main_step_analysis}

\par In this section, the main steps of BARGAIN-MATCH is shown in Algorithm \ref{algo_offMatching}, and the corresponding stability, optimality, and computational complexity are presented. 

\begin{algorithm}[]	
	\label{algo_offMatching}	
	\setlength{\abovecaptionskip}{1pt}   
	\setlength{\belowcaptionskip}{1pt} 
	\SetAlgoLined
	\KwIn{$\mathcal{V}$, $\mathcal{E}$, $\mathbf{T}$}
	\KwOut{$SW$}
	\textbf{Initialization:} 
	Initialize time $t=0$, social welfare $SW=0$, initial positions of vehicles $\mathbf{P}_i, \ i\in \mathcal{V}$, and positions of servers $\mathbf{P}_j, \ j\in \{\mathcal{E}\}$\;
	\While{$t\leq \mathbf{T}$}
	{		
		\textcolor{color}{Create a list $\mathcal{T}^{\text{req}}(t)$ for the tasks that have not been determined where to be offloaded\;
		Call Algorithm \ref{algo_matching} to obtain $\Phi^*(t)$,  ${S}^*_{\text{off}}(t)$, and ${S}^*_{\text{all}}(t)$}\;
		\If{$\Phi \neq \emptyset $}{
			\For{ $j \in \Phi$}{
				\tcp{Resource trading}
				Allcocate resource $f_{j,i}^*(t)$ to $i$\;
				Charge  $c_{j,i}^*(t)$ for per unit of resource on vehicle $i$\;
				\tcp{Task offloading}
				Add the task $k \in \Phi(j)$ to the task list of  server $j$\;
				\tcp{social welfare calculation}
				Calculate the current social welfare $SW(t)$ based on Eq. \eqref{eq_socialwelfare}\;
				\tcp{State update}
				Update the task processing list of server $j$\;
				Update the available computation resources of server $j$\;
			}				
		}
		Calculate the social welfare $SW=SW+SW(t)$\;
		\textcolor{color}{\If{$t \ \% \ T_0 == 0$}{
			Update the mobility of vehicles\;}
		}
		Update time $t = t+\Delta t$\;				
	}
	\Return $SW$\;
	\caption{BARGAIN-MATCH}
\end{algorithm}

\subsubsection{Stability}
\label{sec_stability}

\begin{definition}
	\label{def_block}
	Blocking pair. Assuming that $k\in \mathcal{T}^{\text{req}}(t)$ and $ j\in\{\mathcal{E},o\}$ are not matched with each other under matching $\Phi$, i.e., $i\neq \Phi(j)$ and $j\neq \Phi(i)$, $\Phi$ is blocked by the blocking pair $(i,j)$ if and only if $i$ and $j$ prefer each other to $j^{\prime}=\Phi(i)$ and $i^{\prime}=\Phi(j)$, respectively. 
\end{definition}
\vspace{6pt}
\begin{definition}
	\label{def_stable}
	Stable matching. A matching is stable if and only if there exists no blocking pair \cite{gusfield1989stable}.
\end{definition}
\vspace{6pt}
{\color{color}
\begin{theorem}
	The matching $\Phi$ proposed by this work is stable for every $k\in \mathcal{T}^{\text{req}}(t)$ and $j\in \{\mathcal{E},o\}$.
\end{theorem}

\begin{proof}
	The detailed proof is given in Appendix F of the supplemental material.	
\end{proof}
}
{\color{color}
\subsubsection{Optimality}
\label{sec_optimality}
\begin{theorem}
	The matching $\Phi$ proposed by this work is weak Pareto optimal for each $k\in \mathcal{T}^{\text{req}}$ and $j\in \{\mathcal{E},o\}$.
\end{theorem}

\begin{proof}
	The detailed proof is given in Appendix G of the supplemental material.	
\end{proof}
}

\subsubsection{Complexity Analysis}
\label{sec_complexity}

\begin{theorem}
BARGAIN-MATCH has a polynomial worst-case complexity in each time slot, i.e., $\mathcal{O}\left(\left(|\mathcal{E}|+1\right)\cdot \left(2|\mathcal{V}|+ \min\{|\mathcal{E}|+1, |\mathcal{V}|\}\right)\right)$, where $|\mathcal{V}|$ and $|\mathcal{E}|$ are the number of vehicles and VEC servers, respectively.
\end{theorem}

\begin{proof}
The detailed proof is given is Appendix H of the supplemental material.
\end{proof}

\section{Simulation Results and Analysis}
\label{sec_simulation}

\subsection{Simulation Setup}
\label{simulation_set_up}

\par In this section, the proposed approach is evaluated by simulations implemented in MATLAB(R) 9.9 (R2020b) in a 2.70 GHz Intel Core i7 processor. {\color{color}The road scenario is generated by the \textit{Automated Driving Toolbox 3.2}, where 30 VEC servers are placed on a 10 km 6-lane bidirectional road and 100 vehicles are randomly located on the road initially.} Moreover, the vehicles run at the speed ranges of [2, 30] m/s in either direction. The default values of the simulation parameters are listed in Table \ref{tab_simuParameter}. 

\begin{table}[!hbp]
	\setlength{\abovecaptionskip}{-10pt}%
	\setlength{\belowcaptionskip}{0pt}%
	\caption{Simulation parameters}
	\label{tab_simuParameter}
	\renewcommand*{\arraystretch}{1}
	\begin{center}
		\begin{tabular}{p{.06\textwidth}|p{.21\textwidth}|p{.13\textwidth}}
			\hline
			\hline
			\textbf{Symbol}&\textbf{Meaning}&\textbf{Default value}\\
			\hline
				$\alpha$&CPU parameters& $7.8^{-21}$\cite{Zhang2018}\\
			\hline
				$B_{i,j}$ &Bandwidth between vehicle $i$ and VEC server $j$&40 MHz\cite{Zhang2020a} \\
			\hline
				$\beta^L/\beta^{NL}$&Path loss exponent for LoS/NLoS communication &$3/4$ \cite{zhang2018dense}\\
				\hline
				$c$&Speed of light&$3\times 10^8$ m/s\\
			\hline
				\textcolor{color}{$ C_i^{\max} $}&\textcolor{color}{The budget of vehicle $i$ for the costs payed to the servers }& \textcolor{color}{20 \$}\\
			\hline
				\textcolor{color}{$C_j^{\max} $}&\textcolor{color}{The maximum unit price of server $j$'s computation resources}&\textcolor{color}{1 (\$/GHz)$\cdot f_j^{\max}$} \\
			\hline
				\textcolor{color}{$\mathcal{C}_{i}^{\text{req}}(t)$}&Required computation resources of each bit&[500, 1500] cycles/bit \cite{liu2019cooper}\\ 
			\hline
				$ d_0 $&Reference distance&1 m\\
			\hline
				\textcolor{color}{$\mathcal{D}^{\text{in}}_{i}(t)$}&Task size &$[400, 1000]$
				KB\cite{Zhang2020a}\\
			\hline
				\textcolor{color}{$\mathcal{D}^{\text{out}}_{i}(t)$}&\textcolor{color}{Task result}&\textcolor{color}{$[0.1, 1]$ KB\cite{hui2020game}}\\
			\hline
				\textcolor{color}{$E_i^{\max}$}& \textcolor{color}{Energy constraint of vehicle $i$}&\textcolor{color}{1 (W.h/GHz)$\cdot f_i^{\max}$ \cite{ning2020intelligent}}\\
			\hline
				\textcolor{color}{$E_j^{\max}$}& \textcolor{color}{Energy constraint of server $j$}&\textcolor{color}{1 (W.h/GHz)$\cdot f_j^{\max}$ \cite{ning2020intelligent}}\\
			\hline
				$f_{c}$&Carrier frequency&5.9 GHZ \cite{bazzi2017performance}\\
			\hline
				\textcolor{color}{$f_{i}^{\max}$}&Computation resources of vehicle $i$&[0.5, 1] GHz \cite{lyu2018energy}\\
			\hline
				\textcolor{color}{$f_{j}^{\max}$}&Computation resources of \textcolor{color}{VEC server $j$} &[2, 10] GHz \cite{Zhang2020a}\\ 
			\hline
				\textcolor{color}{$f_{o}^{\max}$}&Computation resources of cloud server $o$ &30 GHz \cite{Wang2022}\\ 
			\hline
				\textcolor{color}{$m^L/m^{NL}$}&\textcolor{color}{Nakagami fading parameter for LoS/NLoS communication} &\textcolor{color}{$2/1$\cite{zhang2018dense}}\\
			\hline
				\textcolor{color}{$N_i^{\text{core}}$}&\textcolor{color}{The CPU core of vehicle $i$}&\textcolor{color}{1}\\
			\hline
				\textcolor{color}{$N_j^{\text{core}}$}&\textcolor{color}{The CPU core of VEC server $j$}&\textcolor{color}{[2, 8]}\\
			\hline
				\textcolor{color}{$N_o$}&\textcolor{color}{Noise power}& \textcolor{color}{-98 dBm}\\
			\hline
				$P_i(t)$&Transmit power&$[-85, 44.8]$ dBm \cite{bazzi2017performance}\\
			\hline
				\textcolor{color}{$r_f$}& \textcolor{color}{Data rate of fiber link} & \textcolor{color}{4 Gb/s \cite{guo2018collaborative1}}\\
			\hline
				\textcolor{color}{$r_{c}$}&\textcolor{color}{Data rate between the edge and cloud}& \textcolor{color}{100 Mb/s \cite{Qian2019}}\\
			\hline
				\textcolor{color}{$\theta^L/\theta^{NL}$}&\textcolor{color}{Standard deviation of shadowing for LoS/NLoS communication}&\textcolor{color}{$3$ dB$/$4 dB\cite{yang2018dense}}\\
			\hline 
				$\tau $&CPU parameters & 3 \cite{Zhang2018}\\		
			\hline
				\textcolor{color}{$T_i^{\max}(t)$}&The maximum permissible delay &[0.1, 5] s \cite{kazmi2021novel}\\
			\hline
				$w_i$/$w_j$&The weight coefficient of vehicle $i$/server $j$&[0, 1] \\
			\hline
		\end{tabular}
	\end{center}
\end{table} 

\par This work evaluates the proposed BARGAIN-MATCH in comparison with three benchmark schemes, i.e., the \textit{entire local offloading (ELO)}, \textit{exhaustive offloading (EXO)}, \textit{nearest \textcolor{color}{VEC} offloading \textcolor{color}{(NVO)}}, and \textcolor{color}{\textit{entire cloud offloading (ECO)}}. Besides, the \textcolor{color}{\textit{non-cooperative game-based offloading (NCO)}\cite{Wang2020} and \textit{One-to-one matching and price-rising-based offloading and resource allocation (OPORA)}} \cite{zhou2019computation} are tailored to be suited to the approach in this paper since there is no feasible solution that can be directly applied to this problem. These approaches are described as follows.

\begin{itemize}
	\item  ELO: all vehicles execute their tasks locally.
	\item  EXO: the tasks of each vehicle are exhaustively offloaded using the optimal offloading strategy.
	\item  NVO: the tasks of each vehicle are offloaded to the nearest VEC server with which the vehicle is currently attached.
	\item  \textcolor{color}{ECO}: all tasks are offloaded to the cloud server.
	\item  \textcolor{color}{NCO}: each vehicle competitively decides the optimal offloading probability by playing a distributed non-cooperative game. Since NCO is designed for the single-VEC server scenario and single-time decision in \cite{Wang2020}, it is adjusted in this paper to adapt to the multi-VEC server and period-time scenario . 
	\item \textcolor{color}{OPORA}: each task of the vehicle is assigned to a VEC server based on the one-to-one matching scheme, and each VEC server is stimulated to allocate the resource using the price-rising scheme.
\end{itemize}

\subsection{System Performance}
\label{sec_performance_utility}

\par In this section, we evaluate the impacts of different system parameters on the performance of social welfare, vehicle utility, and server utility.

\subsubsection{Effect of Time}
\label{sec_effect_time}

\begin{figure*}[!hbt] 
	\centering
	\subfigure[Social welfare]
	{
		\begin{minipage}[t]{0.31\linewidth}
			\centering
			\includegraphics[scale=0.22]{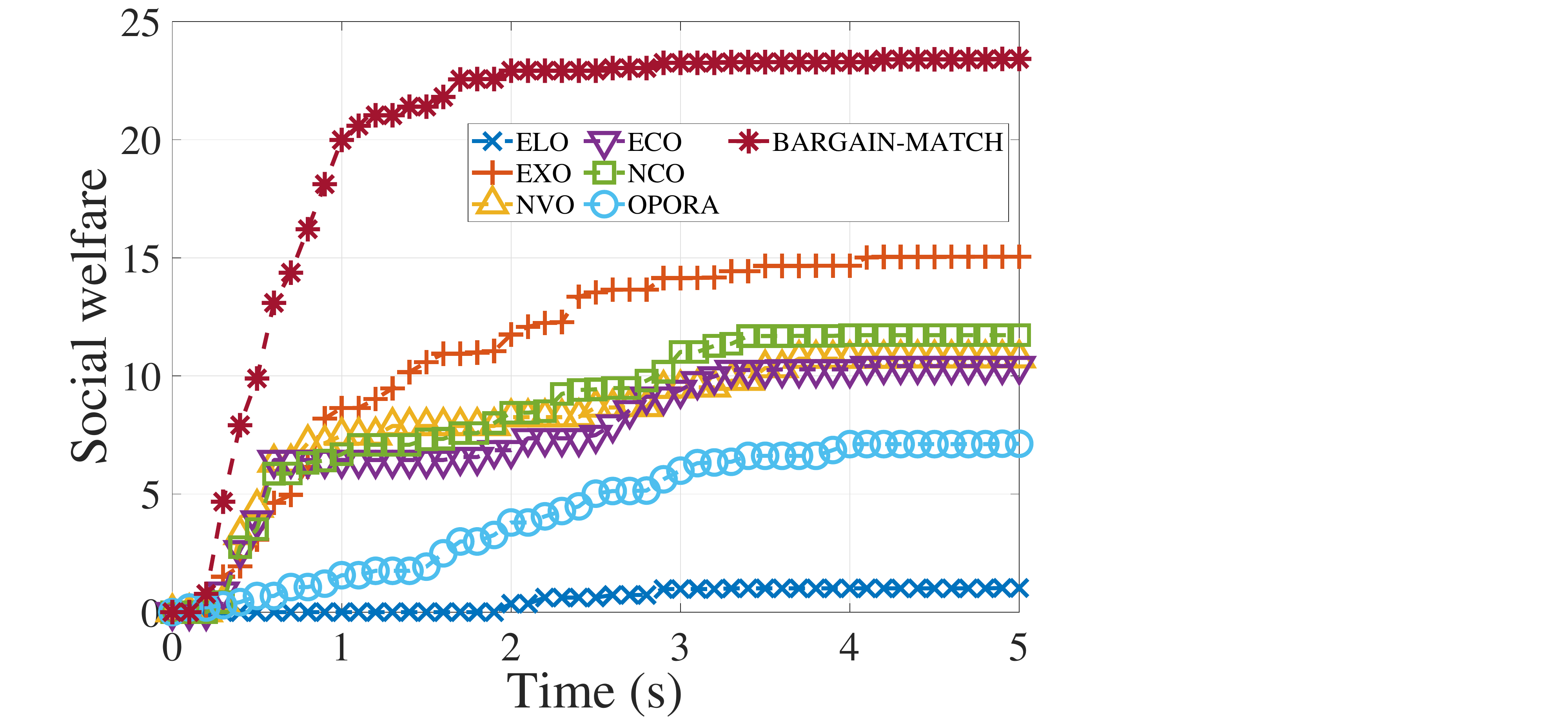}
		\end{minipage}
	}
	\subfigure[The total utility of vehicles]
	{
		\begin{minipage}[t]{0.31\linewidth}
			\centering
			\includegraphics[scale=0.22]{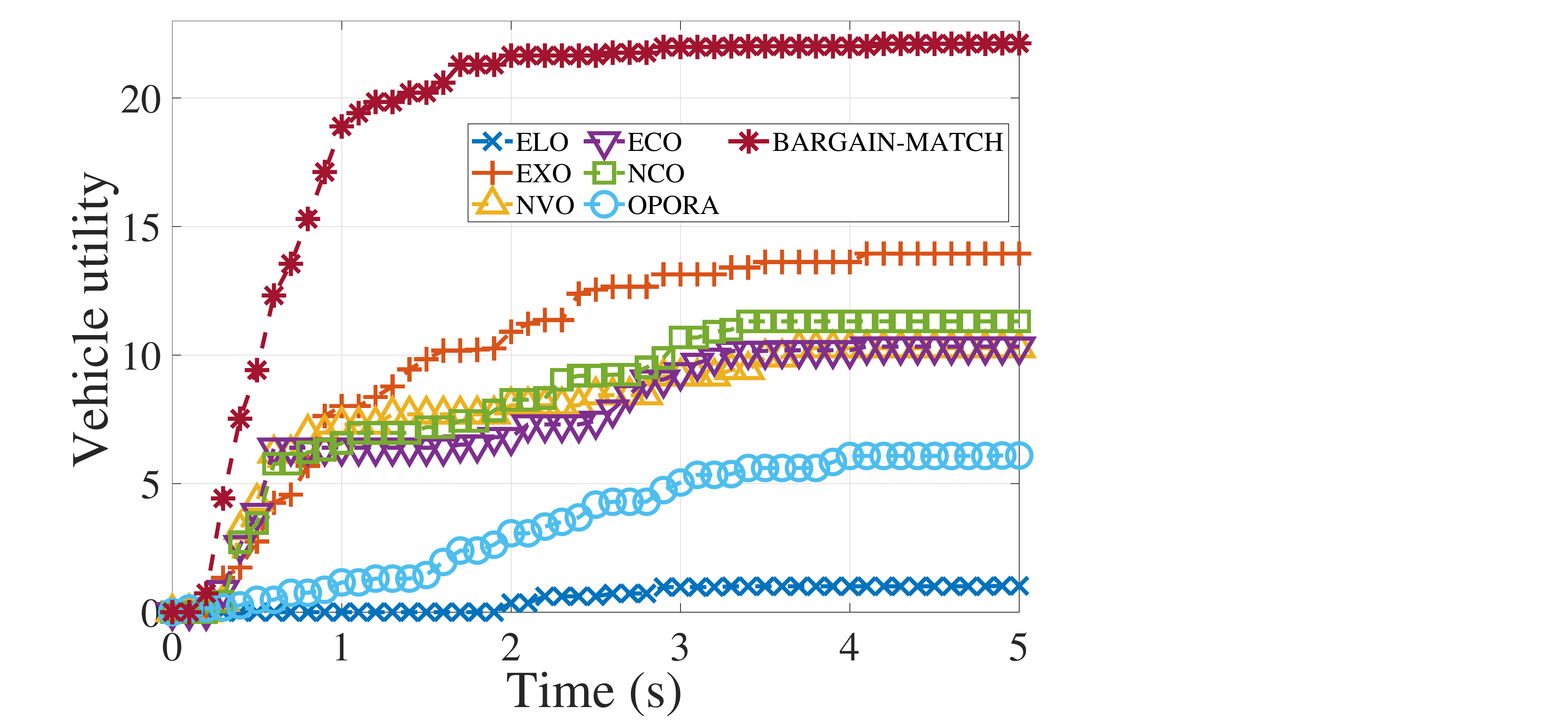}	
		\end{minipage}
	}
	\subfigure[The total utility of servers]
	{
		\begin{minipage}[t]{0.31\linewidth}
			\centering
			\includegraphics[scale=0.22]{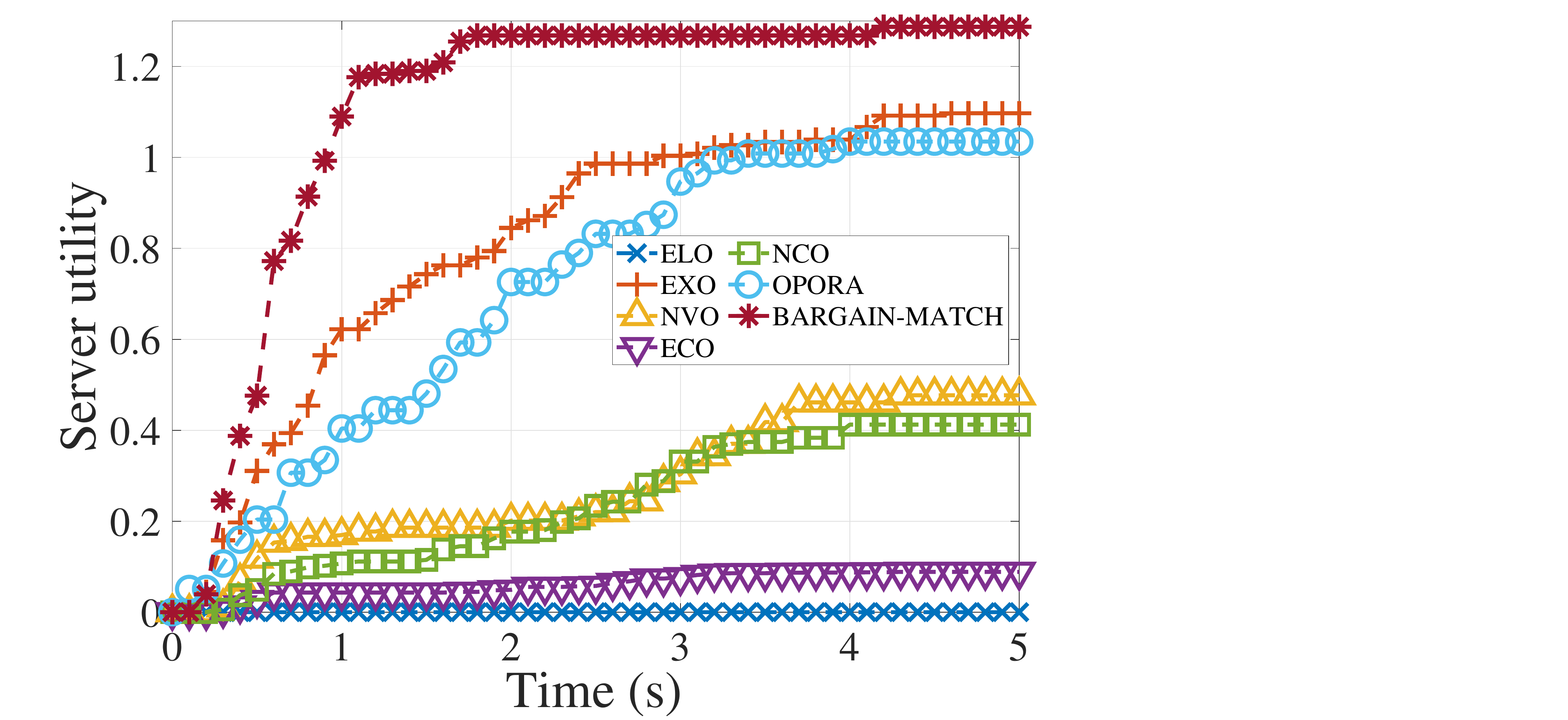}
		\end{minipage}
	}
	\centering
	\caption{\textcolor{color}{System performance with respect to time. (a) Social welfare. (b) Vehicle utility. (c) Server utility.}}
	\label{fig_sw}
	\vspace{-1em}
\end{figure*}

\par Figs. \ref{fig_sw}(a), \ref{fig_sw}(b), and \ref{fig_sw}(c) show the comparative results of the social welfare, the total utility of vehicles, and the total utility of VEC servers among the seven algorithms in terms of time. {\color{color} It can be observed from Fig. \ref{fig_sw} that as time elapses, the social welfare, total utility of vehicles and total utility of servers for all schemes increase. This is because more tasks of vehicles are executed successfully along with time. Moreover, it can be observed that BARGAIN-MATCH outperforms ELO, EMO, NVO, NCO and OPORA in terms of social welfare, vehicle utility and server utility. The reasons are as follows. First, ELO, EMO, NVO and NCO mainly focus on optimizing the task offloading strategy for vehicles but do not consider the resource allocation strategy for VEC servers. The entire local offloading of ELO, the exhaustive offloading of EXO, the nearest VEC server offloading of NVO, the entire cloud offloading of ECO, and the competitive offloading of NCO could lead to congestion and resource over-use at certain vehicles or VEC servers.} Furthermore, although OPORA achieves superior VEC utility than ELO, EMO, NVO and NCO schemes due to the price incentive strategy, it is inferior to BARGAIN-MATCH because it adopts the one-to-one matching strategy and the incentive of random raising pricing, which are less efficient compared to the many-to-one matching strategy and the bargaining incentive. {\color{color} On the one hand, the many-to-one matching of BARGAIN-MATCH improves both the amount of offloaded tasks and the utilization of computation resource by horizontally or vertically offloading the tasks to the VEC servers or cloud server.} On the other hand, the bargaining incentive scheme of BARGAIN-MATCH facilitates the negotiation between the servers and requesting vehicles on the optimal decisions of computation resource allocation and pricing. Consequently, this set of simulation results shows that BARGAIN-MATCH has the overall superior performance on social welfare, vehicle utility, and VEC server utility among the seven algorithms.

\subsubsection{Effect of Vehicle Numbers}
\label{sec_effect_veh_number}

\begin{figure*}[!hbt] 
	\centering
	\subfigure[Social welfare]
	{
		\begin{minipage}[t]{0.31\linewidth}
			\centering
			\includegraphics[scale=0.22]{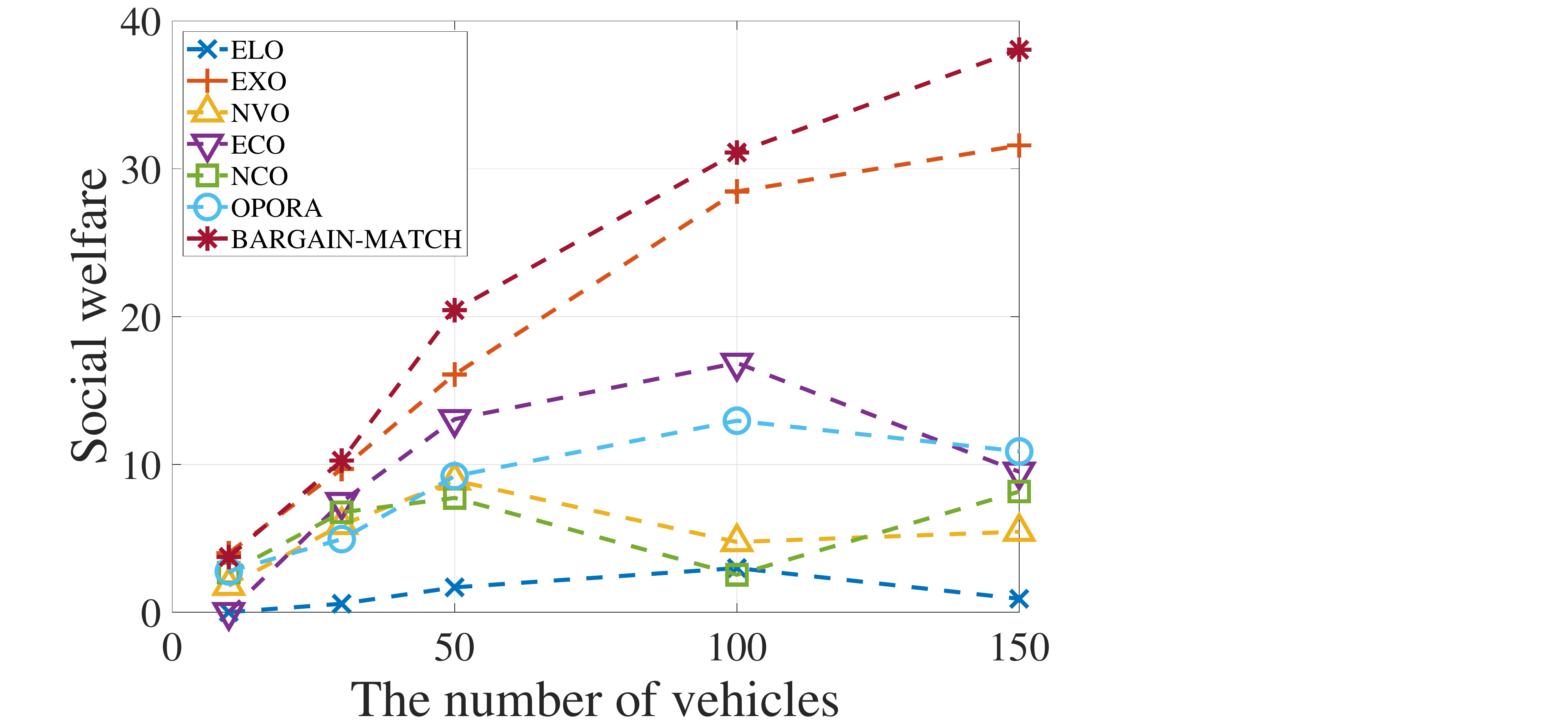}
		\end{minipage}
	}
	\subfigure[The total utility of vehicles]
	{
		\begin{minipage}[t]{0.31\linewidth}
			\centering
			\includegraphics[scale=0.22]{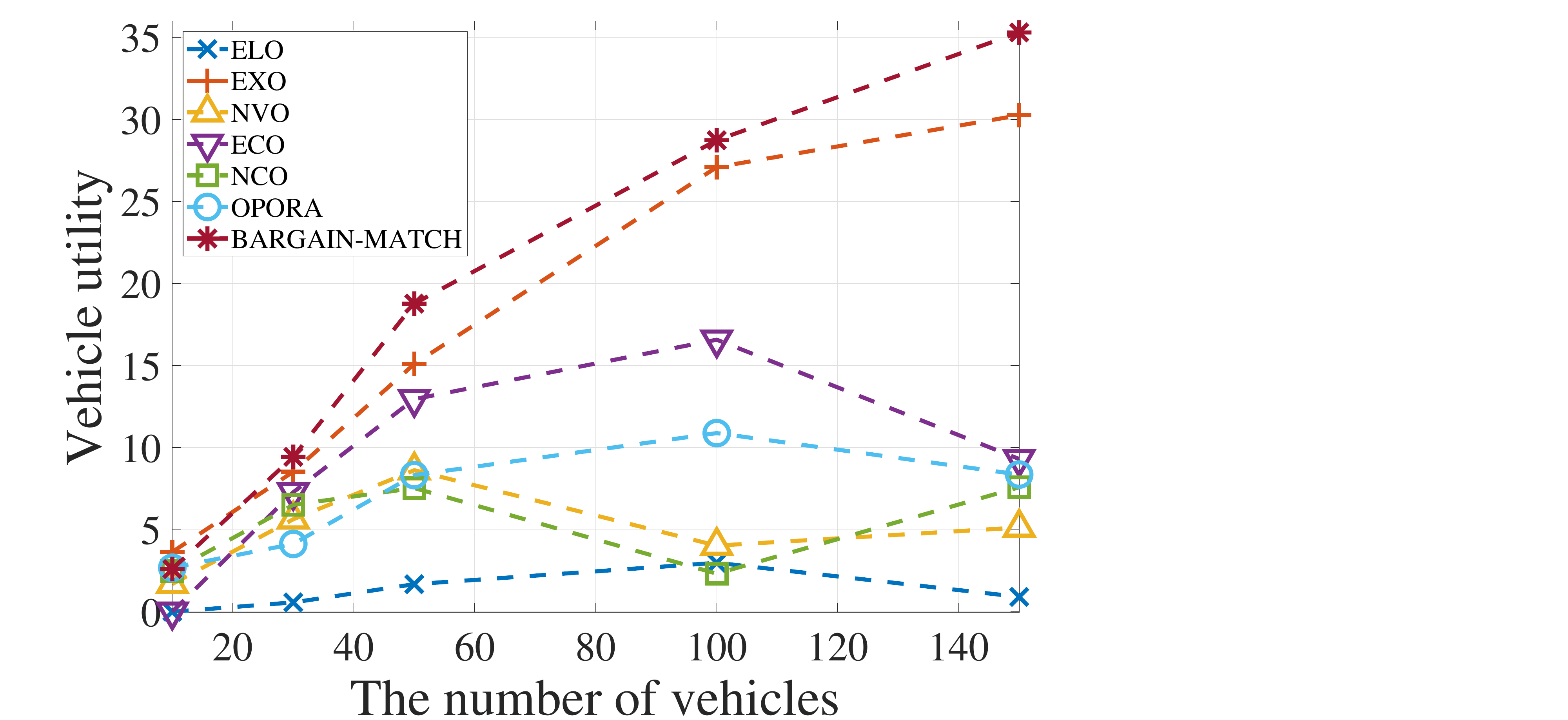}	
		\end{minipage}
	}
	\subfigure[The total utility of servers]
	{
		\begin{minipage}[t]{0.31\linewidth}
			\centering
			\includegraphics[scale=0.22]{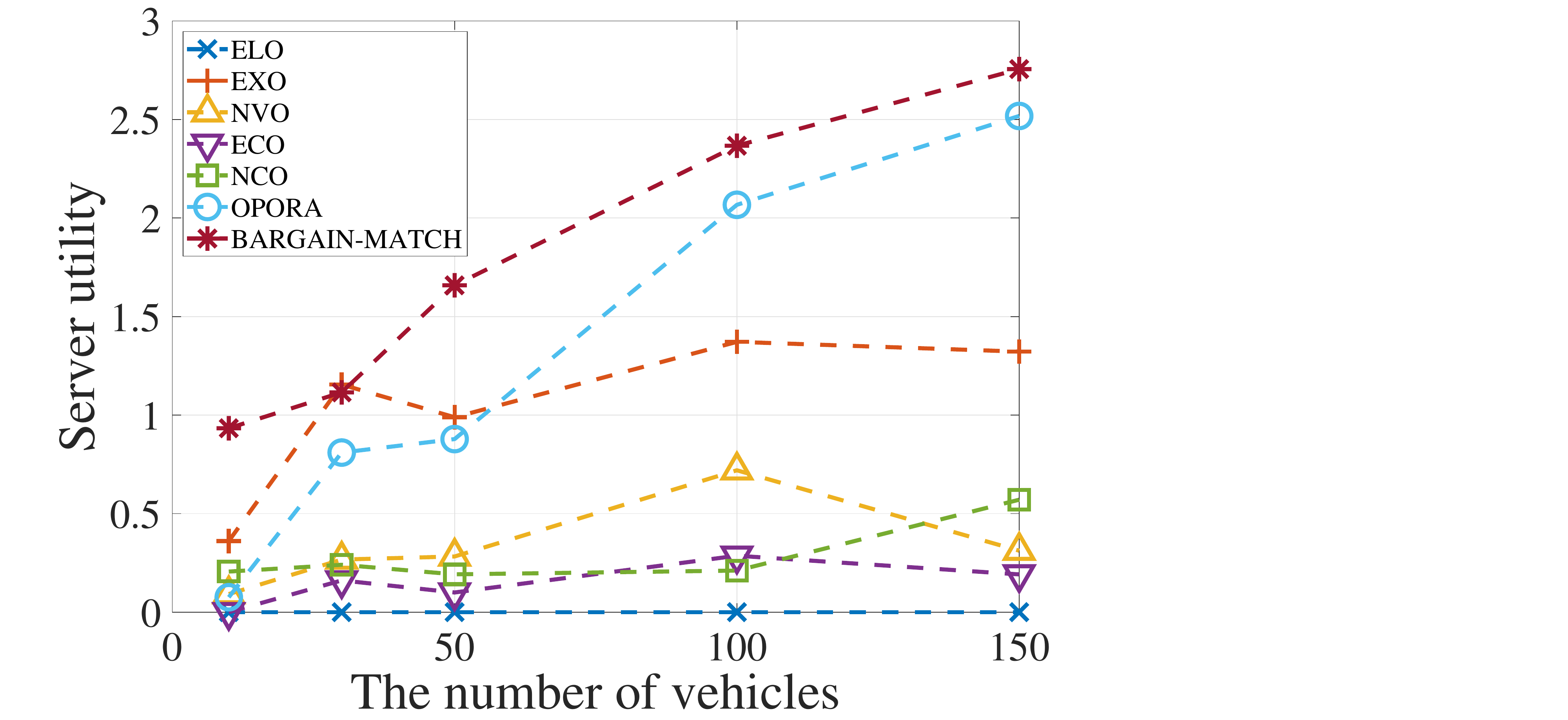}
		\end{minipage}
	}
	\centering
	\caption{System performance with respect to time. (a) Social welfare. (b) Vehicle utility. (c) Server utility.}
	\label{fig_vehicle}
	\vspace{-1em}
\end{figure*}

\par Figs. \ref{fig_vehicle}(a), \ref{fig_vehicle}(b), and \ref{fig_vehicle}(c) compare the impact of the number of vehicles on the performance of social welfare, the total utility of vehicles, and the total utility of servers for different schemes. First, ELO shows the worst performance in terms of social welfare, vehicle utility, and server utility, which is mainly because all tasks are executed locally on vehicles. Furthermore, the social welfare, vehicle utility, and server utility of NVO and ECO exhibit initial upward and then downward tendencies as the number of vehicles grows. This is mainly due to the aggregated amount of tasks with increasing vehicles, which could lead to the possible overload of the nearest VEC servers when adopting EVO and the increasing unfulfilled tasks when adopting ECO. Moreover, with the increase of vehicles, the social welfare and vehicle utility for OPORA show the initial increase and subsequent decrease trends, and the server utility for it rises at a diminishing rate, which is mainly due to the one-to-one offloading strategy and the random increasing price incentive. Additionally, NCO has some random fluctuations in social welfare and vehicle utility and increases slightly in server utility with increasing vehicles. Although the variation tendency of NCO is not statistically significant, it shows inferior performance compared to most of the other schemes (except for ELO). This could be attributed to the probabilistic offloading strategy and the increased competition among vehicles. Besides, EXO shows initial increasing trends in social welfare, vehicle utility, and server utility with the increasing number of vehicles, and the trends slow down gradually or exhibit a decreasing trend. As explained before, although tasks can be offloaded to the optimal server, the exhaustive offloading strategy of EXO could lead to possible congestion or resource shortage for certain VEC servers, as the vehicles continuously increase. Last, it can be observed that BARGAIN-MATCH exhibits progressively increasing trends in the performance of social welfare, vehicle utility, and server utility, maintaining a relatively superior level among the seven schemes. The set of simulation results indicates the better scalability of the proposed BARGAIN-MATCH with an increasing number of vehicles.

{\color{color1}
\subsubsection{Effect of Vehicle Speed}
\label{sec_effect_mob}
\begin{figure*}[!hbt] 
	\centering
	\subfigure[\textcolor{color1}{Social welfare}]
	{
		\begin{minipage}[t]{0.31\linewidth}
			\raggedleft
			\includegraphics[scale=0.22]{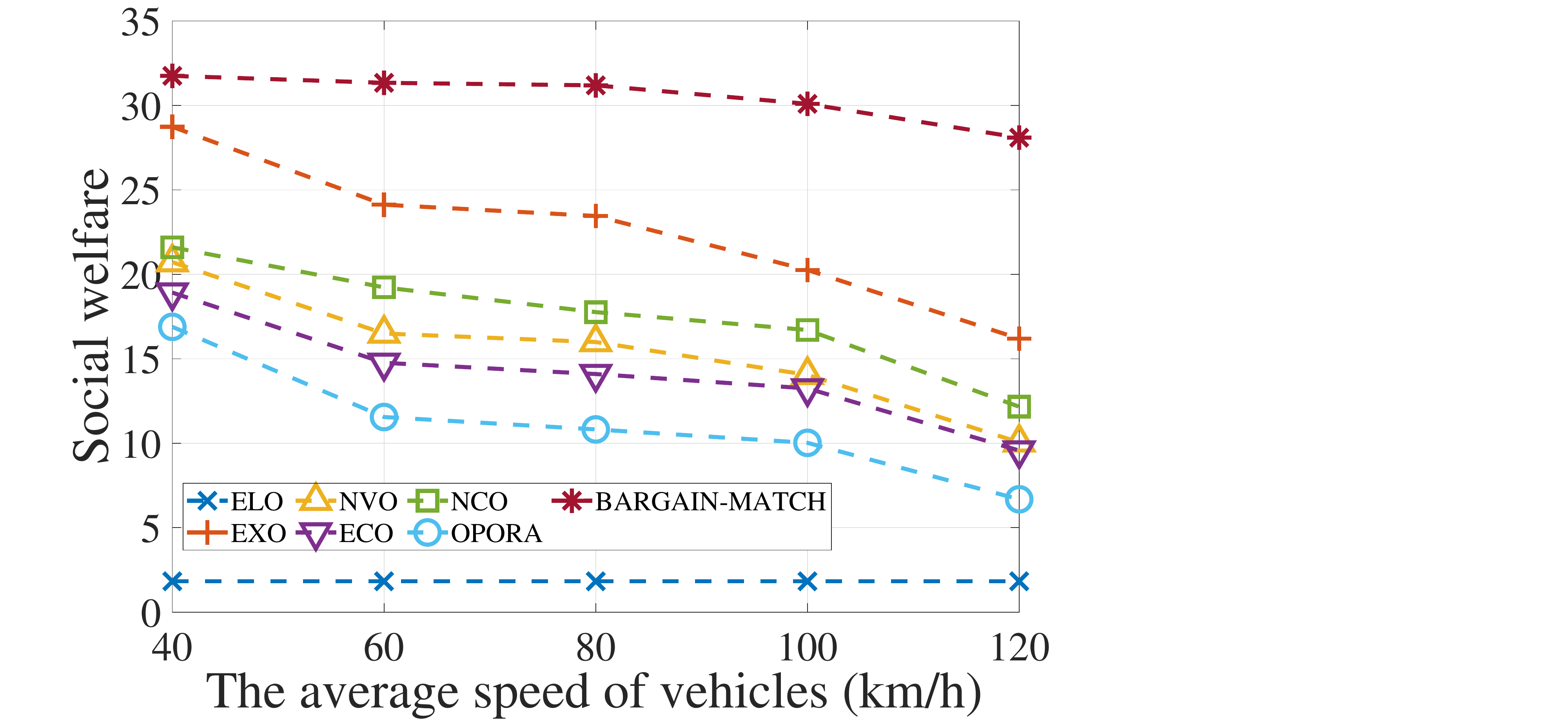}
		\end{minipage}
	}
	\subfigure[\textcolor{color1}{The total utility of vehicles}]
	{
		\begin{minipage}[t]{0.31\linewidth}
			\centering
			\includegraphics[scale=0.22]{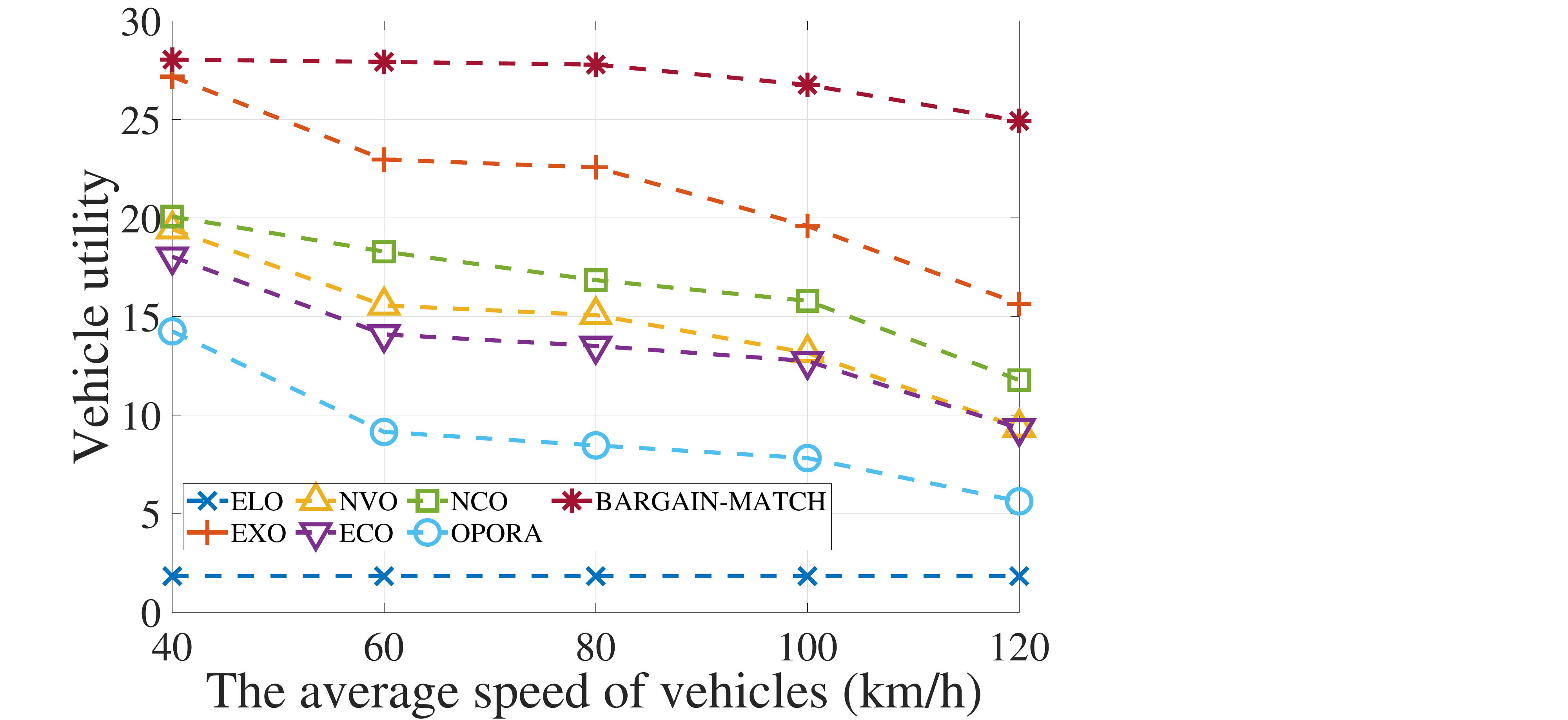}	
		\end{minipage}
}
	\subfigure[\textcolor{color1}{The total utility of servers}]
	{
		\begin{minipage}[t]{0.31\linewidth}
			\raggedright
			\includegraphics[scale=0.22]{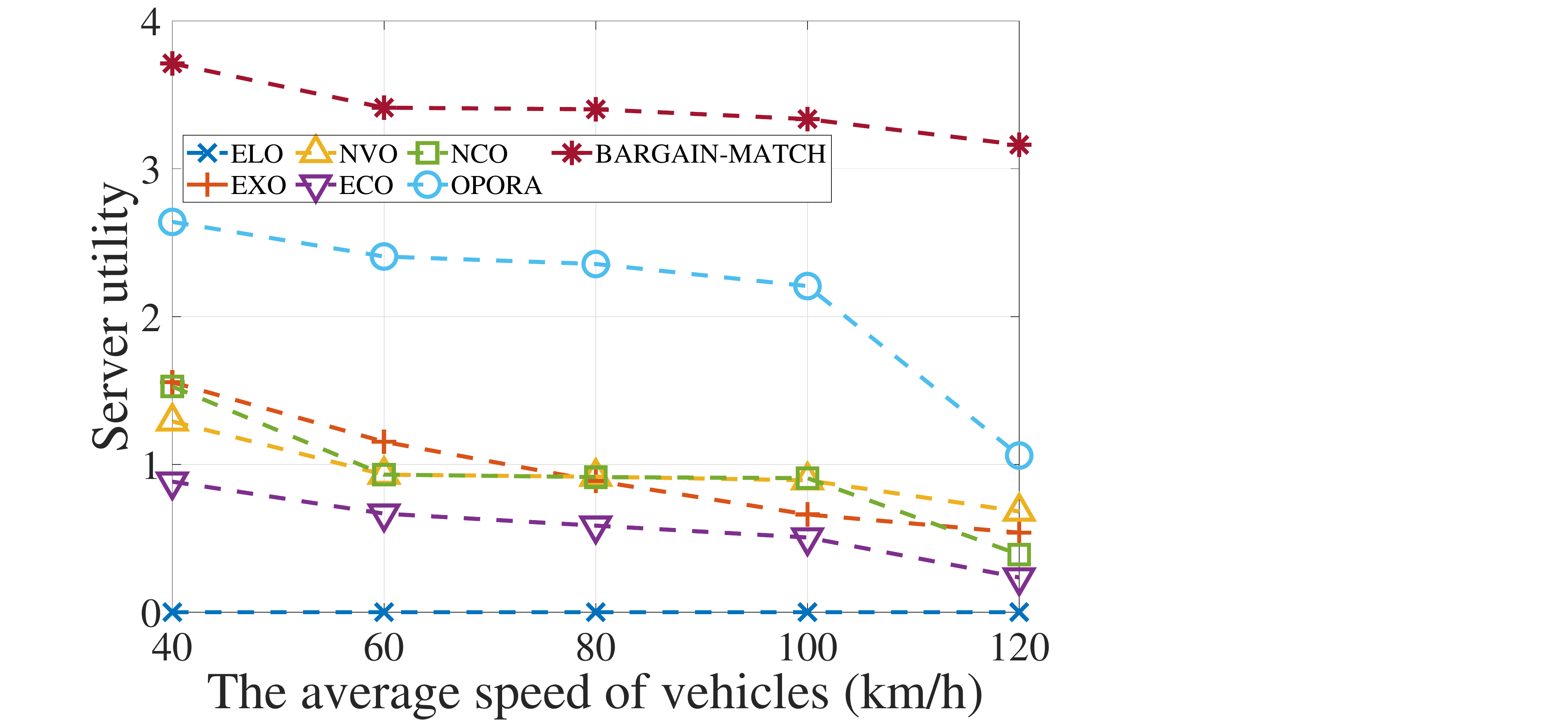}
		\end{minipage}
	}
	\centering
	\caption{\textcolor{color1}{System performance with respect the average speed of vehicles. (a) Social welfare. (b) Vehicle utility. (c) Server utility.}}
	\label{fig_mobi}
	\vspace{-1em}
\end{figure*}

\par Figs. \ref{fig_mobi}(a), \ref{fig_mobi}(b), and \ref{fig_mobi}(c) compare the social welfare, the total utility of vehicles, and the total utility of servers among the comparative algorithms, respectively. First, it can be observed from Fig. \ref{fig_mobi} that the ELO shows the invariant but the worst performance in terms of social welfare, vehicle utility, and server utility with respect to the average speed of vehicles. This is obvious since the tasks of vehicles are executed locally without communicating with the VEC servers. Furthermore, for the algorithms of EXO, NVO, ECO, NCO, OPORA and BARGAIN-MATCH, the curves of social welfare, vehicle utility, and server utility show overall downward trends with the increasing speed of vehicles. This is mainly because the high mobility of vehicles indicates that the vehicles could move out of the service range of the VEC server during the task uploading or task computing more frequently, leading to more repetitive handovers or even service interruptions. Therefore, the high mobility of vehicles could cause increased service delay and more task failures, which further results in the degraded satisfaction level of vehicles and the decreased revenues of servers. Specifically, the performances of EXO, NVO, ECO, NCO, and OPORA degrade significantly when the average speed exceeds 100 km/h. Finally, it can be observed that the performances of the proposed BARGAIN-MATCH exhibit relatively steady downward trends as the average speed of vehicles increases, maintaining superior levels among the seven schemes. The reason is that the proposed BARGAIN-MATCH performs cooperative decisions of task offloading by considering both horizontal and vertical task migrations, where the handover delay is incorporated to improve the connectivity for moving vehicles. In conclusion, this set of simulation results demonstrates that the proposed BARGAIN-MATCH can achieve relatively stable and superior performance in terms the social welfare, vehicle utility, and server utility against varying speeds of vehicles.
}

\subsubsection{Effect of Task Size}
\label{sec_effect_task}

\par The effect of the initial price of the computation resources on the system performance for the comparative algorithms is presented in Appendix I.1.1 of the supplementary material due to the page limitation.

\subsubsection{Effect of the Computation Resources of VEC Servers}
\label{sec_effect_resource}

\par The effect of the initial price of the computation resources on the system performance for the comparative algorithms is presented in Appendix I.1.2 of the supplementary material due to the page limitation.

\subsubsection{Effect of the Initial Price of Computation Resources}
\label{appen_effect_price}

\par The effect of the initial price of the computation resources on the system performance for the comparative algorithms is presented in Appendix I.1.3 of the supplementary material due to the page limitation.

\subsection{System Efficiency}
{\color{color1} 
	

\par For performance evaluation, the following statistics are collected: the generation time of each task  $t_{i}^{\text{req}}(t),\forall i\in \mathcal{V},t\in \mathbf{T}$; the successful completion time of each task  $t_{i}^{\text{com}}(t),\text{ }\forall i\in \mathcal{V},t\in \mathbf{T}$; the number of successfully completed tasks during the considered timeline, which is denoted as $N^\text{succ}$. Based on these statistics, the following performance metrics are defined.

\begin{itemize}
	\item \textit{Average processing rate} (APR) is defined as the average amount of task (in bits) that is processed per unit time, which is given as:
	\begin{equation} 
		\text{APR}=\frac{{\sum_{t\in \mathbf{T}}\sum_{i\in \mathcal{V}}}p_i^{\text{gen}}(t)\cdot\mathcal{C}_i^{\text{req}}(t)}{\sum_{t\in \mathbf{T}}\sum_{i\in \mathcal{V}}p_i^{\text{gen}}(t)\cdot \left(t_i^{\text{cmp}}(t) - t_i^{\text{req}}(t)\right)}.
	\end{equation}

	
	\item \textit{Average completion delay} (ACD) is defined as the average delay of completing a task successfully: 
	\begin{equation} 
		\text{ACD}=\frac{\sum_{t\in \mathbf{T}}\sum_{i\in \mathcal{V}}p_i^{\text{gen}}(t)\cdot\left(t_i^{\text{cmp}}(t) - t_i^{\text{req}}(t)\right)}{\sum_{t\in \mathbf{T}}\sum_{i\in \mathcal{V}}p_i^{\text{gen}}(t)}.
	\end{equation}
	
	
	\item \textit{Average completion ratio} (ACR) is defined as the ratio of tasks that are successfully completed to the total number of tasks generated during the considered duration, which is as follows:
	\begin{sequation} 
			\text{ACR}=\frac{N^\text{succ}}{\sum_{t\in \mathbf{T}}\sum_{i\in \mathcal{V}}p_i^{\text{gen}}(t)}.
	\end{sequation}
	

\end{itemize}
}

\subsubsection{Effect of Task Size}

\begin{figure*}[!hbt] 
	\centering
	\subfigure[APR]
	{
		\begin{minipage}[t]{0.31\linewidth}
			\centering
			\includegraphics[scale=0.22]{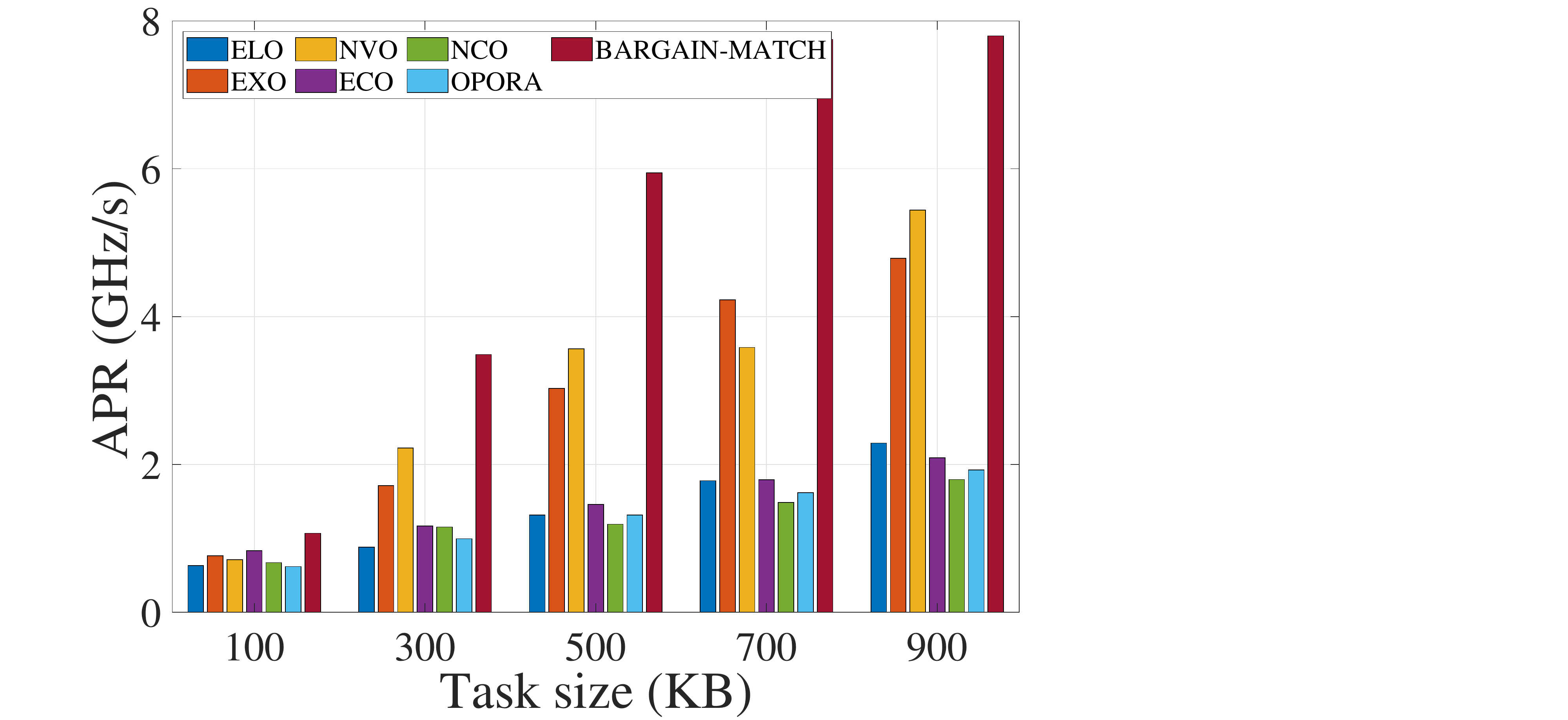}
		\end{minipage}
	}
	\subfigure[ACD]
	{
		\begin{minipage}[t]{0.31\linewidth}
			\centering
			\includegraphics[scale=0.22]{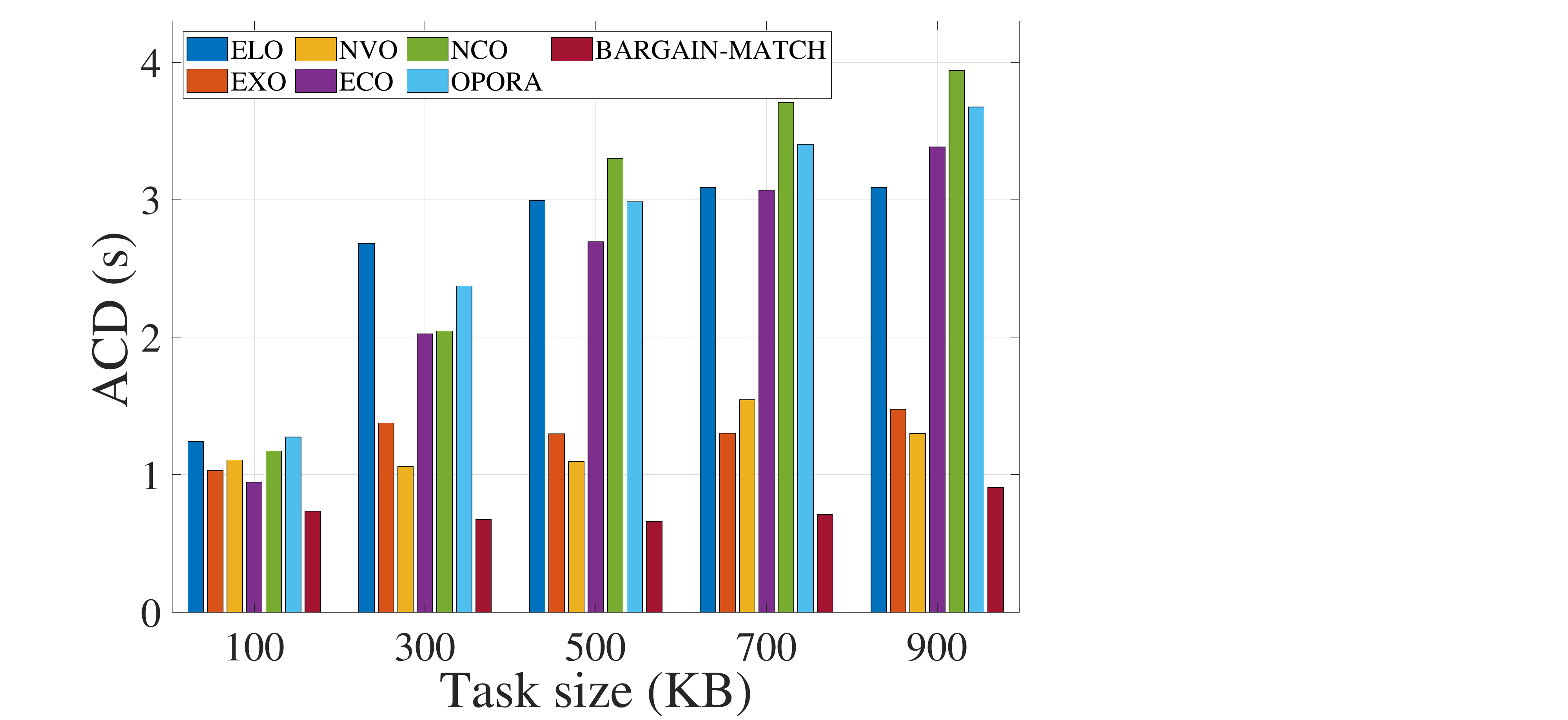}	
		\end{minipage}
	}
	\subfigure[\textcolor{color1}{ACR}]
	{
		\begin{minipage}[t]{0.31\linewidth}
			\centering
			\includegraphics[scale=0.22]{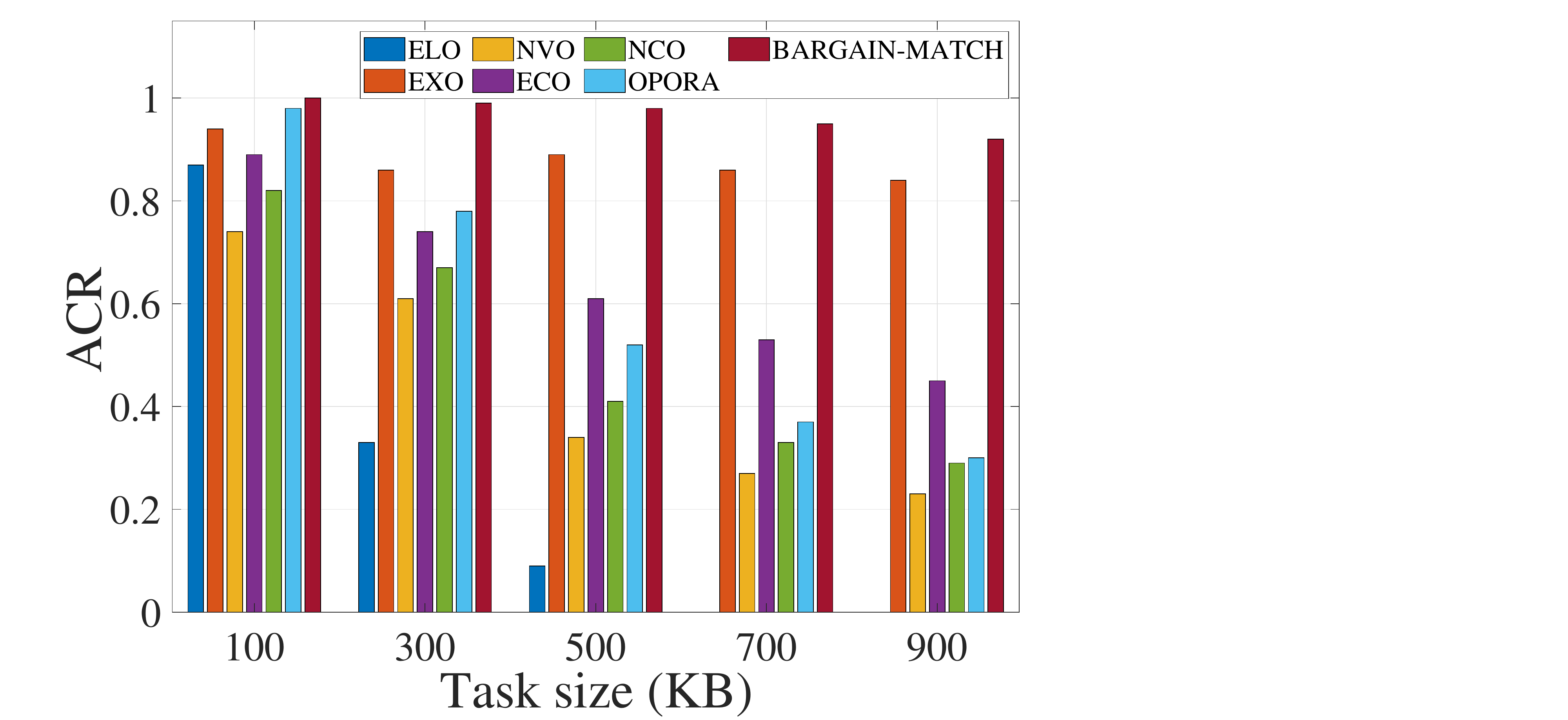}	
		\end{minipage}
	}
	\centering
	\caption{System efficiency with respect to the average size of tasks. (a) APR. (b) ACD. \textcolor{color1}{(c) ACR.}}
    \label{fig_efficiency_task}
	\vspace{-1em}
\end{figure*}

{\color{color1} 
\par  Figs. \ref{fig_efficiency_task}(a), \ref{fig_efficiency_task}(b), and \ref{fig_efficiency_task}(c) compare the APR, ACD, and ACR of the seven algorithms under different task sizes, respectively. First, both the APR and ACD of ELO, EXO, NVO, ECO, and NCO show overall upward trends with the increasing of task size, and the ACR of them show the opposite trends with the task size increases. Obviously, this is because the workloads of vehicles or servers become heavier with the increasing of task sizes, leading to the increased amount of task processing per unit time, the increased delay of task completion, and the decreased task completion ratio. Besides, it can be observed that the proposed BARGAIN-MATCH shows a significant rising trend in APR, a slight upward trend in ACD, and a slight downward trend in ACR with the increasing of task size. This implies that the processing rate of BARGAIN-MATCH increases significantly with relatively low costs of delay and task failure as the workload increases. Furthermore, BARGAIN-MATCH achieves the highest APR, the lowest ACD, and the highest ACR compared to the other schemes with the increasing of task size. BARGAIN-MATCH tries to stimulate cooperation among servers for inter-server task offloading and cooperation between servers and vehicles for intra-server resource allocation according to the varying works and the available resources of servers. In conclusion, the result set in Fig. \ref{fig_efficiency_task} demonstrates the efficiency of BARGAIN-MATCH in terms of the APR, ACD, and ACR under varying task sizes.
}

\subsubsection{Effect of the Computation Resources of VEC Servers}
{\color{color1}
	\par The effect of the computation resources on the system efficiency for the comparative algorithms is given in Appendix I.2.1 of the supplementary material due to the page limitation.
}

\subsubsection{Effect of the Initial Price of Computation Resources}
{\color{color1}
	\label{sec_effect_initial_price}
	\par The effect of the initial price of the computation resources on the system efficiency for the comparative algorithms is given in Appendix I.2.2 of the supplementary material due to the page limitation.
}

\begin{figure}[!hbt]
	\setlength{\abovecaptionskip}{-20pt}%
	\setlength{\belowcaptionskip}{-20pt}%
	\centering
	\includegraphics[scale=0.28]{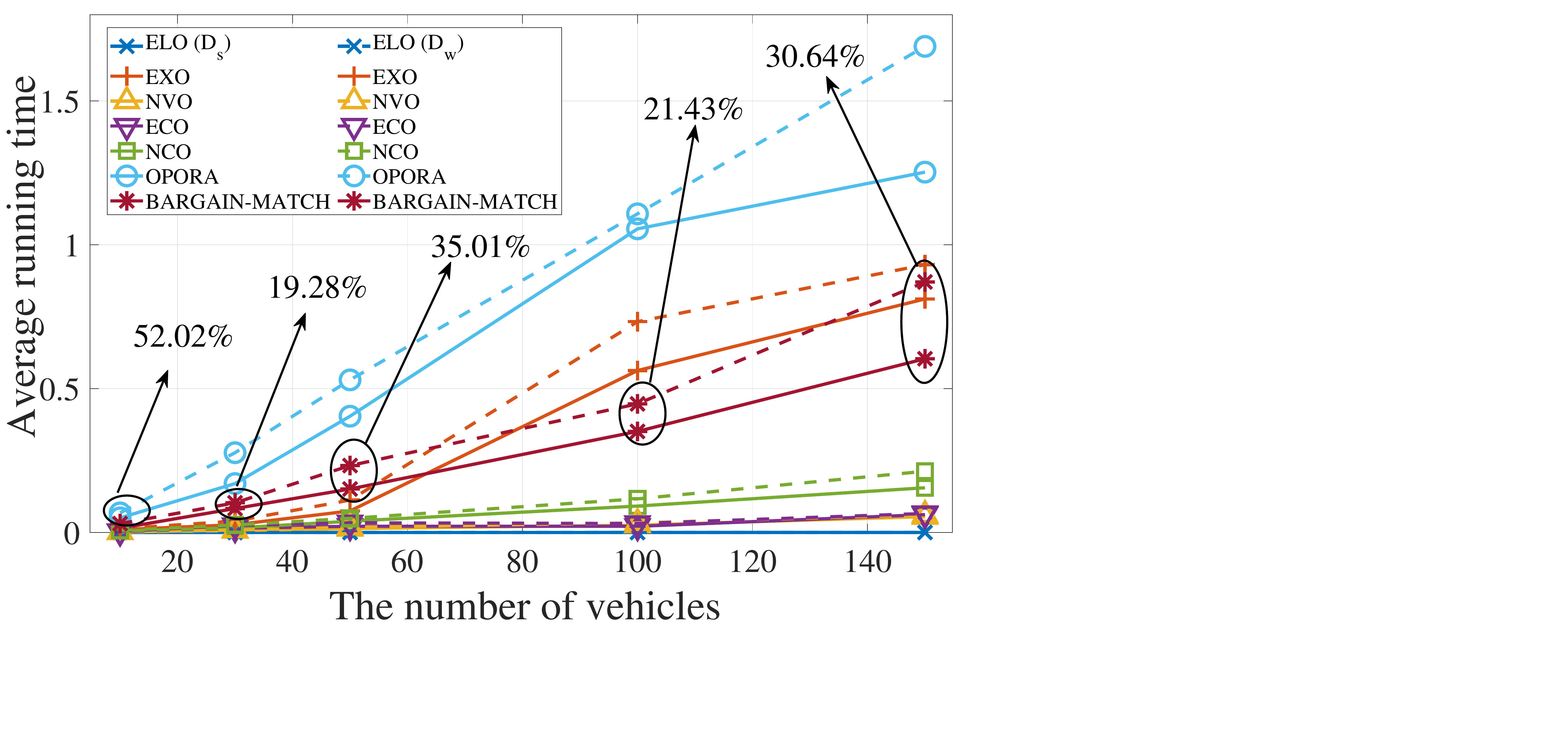}
	\caption{Average running time with respect to the number of vehicles.}
	\label{fig_run}	
    \vspace{-1em}
\end{figure}

%
%
{ \color{color}
\subsection{Algorithm Running Time}

\par To evaluate the execution time of different approaches, Fig. \ref{fig_run} shows the average running time versus the number of vehicles for the seven algorithms. To show the impact of the processors, the proposed approach is implemented on a relative strong device ($D_s$) equipped with Intel Core i7-12700H, 2.70 GHz processor, 16.0 GB RAM memory and on a relative weak device ($D_w$) equipped with Intel Core i7-9750H, 2.60GHz processor, 16.0 GB RAM memory, respectively. From the perspective of the comparative approaches, it can be observed from Fig. \ref{fig_run} that the average running time increases for each algorithm with increasing number of vehicles. Specifically, the algorithms of ELO, NVO, ECO, and NCO have lower time complexity compared to EXO, OPOPRA, and BARGAIN-MATCH. This is because these approaches make offloading decisions directly, which however have relative inferior system performance and efficiency compared to the other approaches, as shown in Fig. \ref{fig_sw} to Fig. \ref{fig_efficiency_task}. Furthermore, OPOPRA shows significantly higher time complexity among the seven schemes, which is mainly attributed to the time-consuming strategies of one-to-one matching and the random price rising incentive. Moreover, it can be observed that the average execution time of BARGAIN-MATCH increases linearly with increasing number of vehicles, which is consistent with the theoretical analysis. Besides, EXO takes less execution time then BARGAIN-MATCH when the number of vehicles is relative small ($\leq 75$), while it is more time-consuming when the network becomes denser, and the difference enlarges gradually. From the perspective of the capability of devices, it can be observed that the average running time decreases approximately 19.28\% to 52.02\% when the simulation runs on $D_s$. Concluding from Fig. \ref{fig_run}, BARGAIN-MATCH achieves superior performance in social welfare, vehicle utility, and server utility with the time complexity higher than the schemes that adopt direct decision making and lower than EXO in relative dense networks. Furthermore, the proposed approach can be completed in polynomial time with linearly increased complexity over the number of vehicles. Besides, the execution time of the proposed approach could be further reduced when it runs in the real VEC node with stronger processing capability than the device we used.
}

{\color{color}
\section{Discussion}
\subsection{The Case of Vehicular Applications}
\label{sec_discussion}

\par Based on the standards of 5G Automotive Association (5GAA) \cite{5g2019c} and ETSI MEC \cite{isg2018multi,Spinelli2021}, we consider the following vehicular applications, i.e., 
i) \textit{vehicle collision warning}, ii) \textit{emergency break warning}, iii) \textit{traffic jam warning}, iv) \textit{hazardous location warning}, and v) \textit{speed harmonization}. The characteristics of each application that are mapped to the task model in Section \ref{sec_system_model} are given as follows. First, the task size and maximum acceptable delay for applications i)-v) are given as: \textbf{i)} [300, 1000] B and 100 ms,
\textbf{ii)} [200, 400] B and 120 ms, \textbf{iii)} 300 B and 2000 ms, \textbf{iv)} [300, 1000] B and [1000, 2000] ms (safety) or $[10^4, 2\times10^5]$ ms (route obstruction), and \textbf{v)} [300, 1000] B and [400, 1500] ms, respectively \cite{5g2019c}. Furthermore, the computational intensity and the result of the above applications are given as $[10^3, 10^4]$ cycles/bit and [0.1, 1] KB, respectively \cite{Spinelli2021}. We evaluate the performance of the proposed approach for these vehicular applications in Appendix J of the supplementary material.
}

{\color{color1}
\subsection{The Impact of Multiple Access Schemes} 
\label{sec_discussion_OFDMA}

\par In this sub-section, we discuss the impact of the employed multiple access schemes on the performance. Specifically, in Appendix K of the supplementary martial, we evaluate the performance of the proposed BARGAIN-MATCH in the scenario where the network employs the OFDMA, followed by the discussion on the  extensibility of the proposed approach for more complicated scenarios. 

}

%
%

\section{Conclusion}
\label{sec_conclusion}

\par In this work, we investigate the computation allocation and task offloading for VEC servers and vehicles in VEC networks. \textcolor{color}{First, to coordinate the space-time-requirement heterogeneity among tasks and the computational heterogeneity among servers, this work employs a hierarchical framework where the intra-server resource allocation and inter-server offloading are decided through the horizontal and vertical collaboration among vehicle, edge, and cloud layers under the coordination of the controller.} \textcolor{color}{Furthermore, JRATOP is formulated to maximize the system utility by jointly optimizing the strategies of resource allocation, resource pricing, and task offloading. To solve the NP-hard problem, we propose the BARGAIN-MATCH that consists of the bargaining-based trading model for intra-server resource allocation and a matching-based collaboration approach for inter-server task offloading.} Besides, the proposed BARGAIN-MATCH is proved to be stable, weak Pareto optimal, and polynomial complex. Simulation results demonstrate that BARGAIN-MATCH achieves superior performance in terms of the system utility, vehicle utility and server utility compared to the conventional approaches. Moreover, it can improve the task processing rate and task processing delay significantly, especially when the system workload is heavy.


\section*{Acknowledgment}

\par This work was supported in part by the National Natural Science Foundation of China under Grants 62172186, 62002133, 61872158, and 62272194, in part by the Science and Technology Development Plan Project of Jilin Province under Grants 20210101183JC and 20210201072GX, and in part by the Young Science and Technology Talent Lift Project of Jilin Province under Grant QT202013.

\ifCLASSOPTIONcaptionsoff
\newpage
\fi

\bibliographystyle{IEEEtran}
\bibliography{references.bib}

\vspace{-23pt}
\begin{IEEEbiography}[{\includegraphics[width=1in,height=1.25in,clip,keepaspectratio]{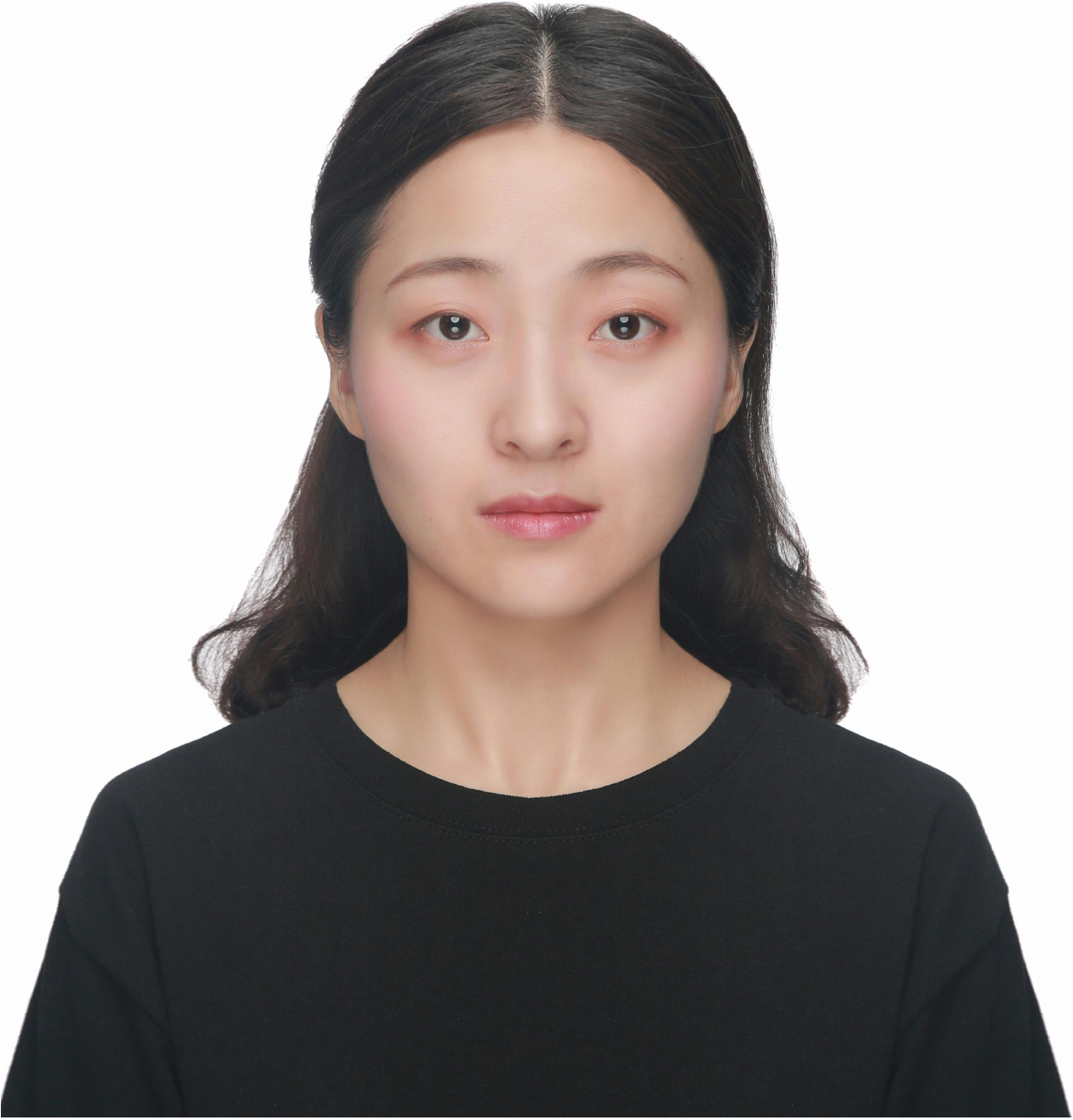}}]{Zemin Sun} (S'21)  received a BS degree in Software Engineering,  an MS degree and a Ph.D degree  in Computer Science and Technology from Jilin University, Changchun, China, in 2015, 2018, and 2022, respectively. Her research interests include vehicular networks, edge computing, and game theory. 
\end{IEEEbiography}

\vspace{-23pt}
\begin{IEEEbiography}[{\includegraphics[width=1in,height=1.25in,clip,keepaspectratio]{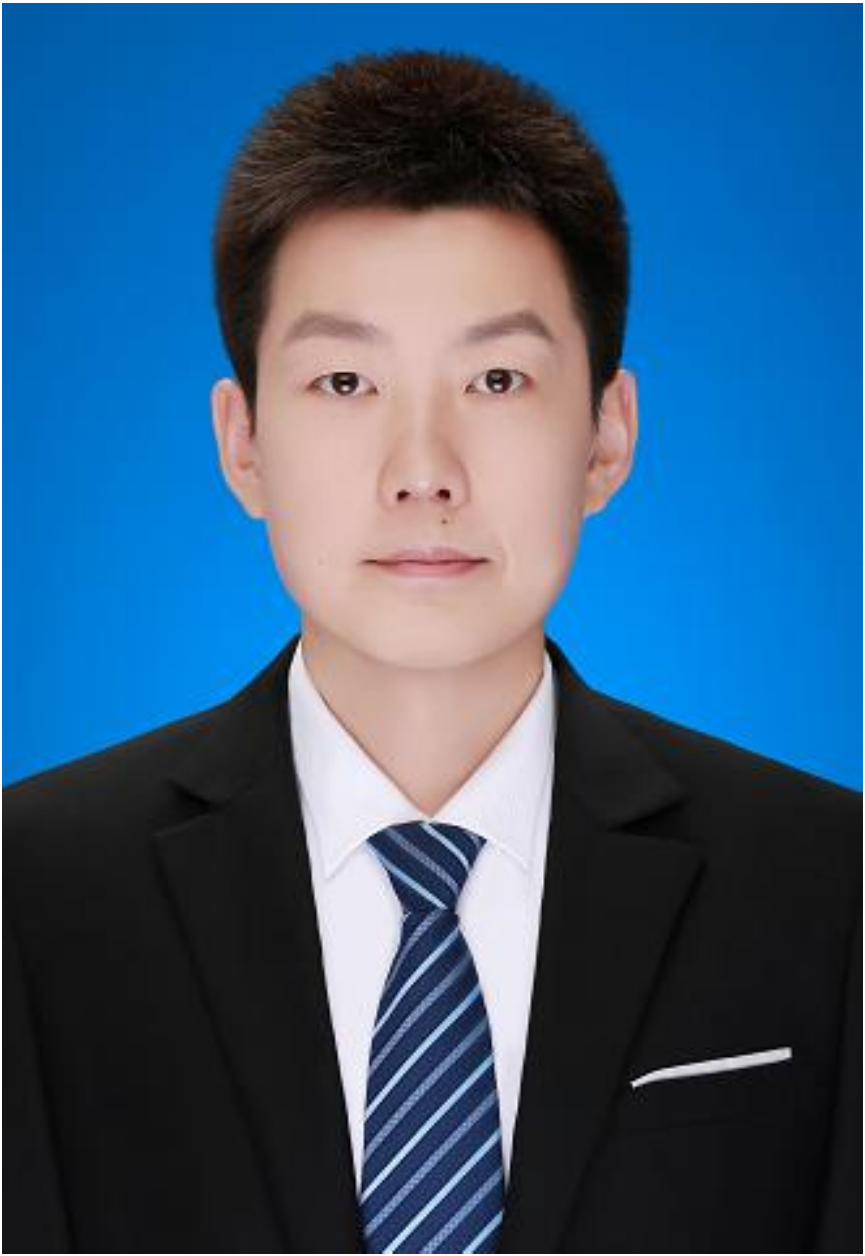}}]{Geng Sun} (S'17-M'19) received the B.S. degree in communication engineering from Dalian Polytechnic University, and the Ph.D. degree in computer science and technology from Jilin University, in 2011 and 2018, respectively. He was a Visiting Researcher with the School of Electrical and Computer Engineering, Georgia Institute of Technology, USA. He is an Associate Professor in College of Computer Science and Technology at Jilin University, and His research interests include wireless networks, UAV communications, collaborative beamforming and optimizations.
\end{IEEEbiography}

\vspace{-23pt}
\begin{IEEEbiography}[{\includegraphics[width=1in,height=1.25in,clip,keepaspectratio]{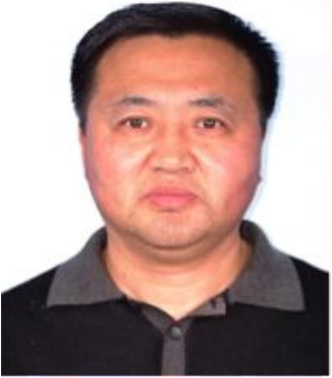}}]{Yanheng Liu} received the M.Sc. and Ph.D. degrees in computer science from Jilin University, People's Republic of China. He is currently a professor in Jilin University, People's Republic of China. His primary research interests are in network security, network management, mobile computing network theory and applications, etc. He has co-authored over 90 research publications in peer reviewed journals and international conference proceedings of which one has won ``best paper" awards. Prior to joining Jilin University, he was visiting scholar with University of Hull, England, University of British Columbia, Canada and Alberta University, Canada.
\end{IEEEbiography}
\vspace{-23pt}

\begin{IEEEbiography}[{\includegraphics[width=1in,height=1.25in,clip,keepaspectratio]{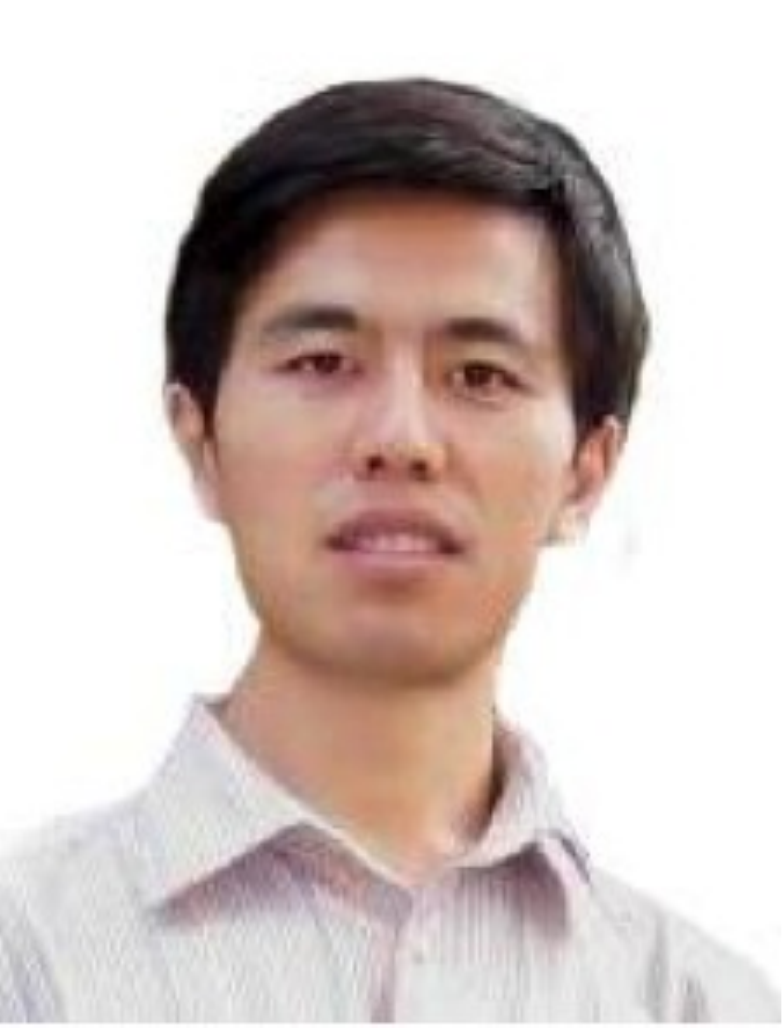}}]{Jian Wang} received the B.Sc., M.Sc., and Ph.D. degrees in computer science from Jilin University,	Changchun, China, in 2004, 2007, and 2011, respectively. He is currently a Professor with the College of Computer Science and Technology, Jilin University. He is interested in topics related to	wireless communication and vehicular networks, especially for network security and privacy protection. He has published over 40 articles in international journals.
\end{IEEEbiography}

\vspace{-23pt}

\begin{IEEEbiography}[{\includegraphics[width=4in,height=1.25in,clip,keepaspectratio]{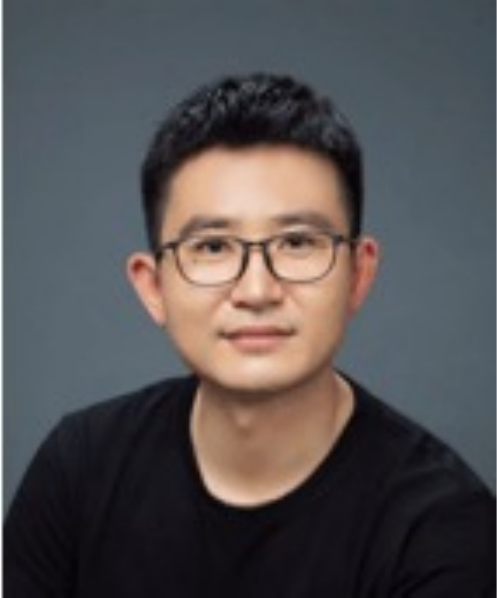}}]{Dongpu Cao} (M'08) received the Ph.D. degree from Concordia University, Canada, in 2008. He is a Professor at Tsinghua University. His current research focuses on driver cognition, automated driving and social cognitive autonomous driving. He has contributed more than 200 papers and 3 books. He received the SAE Arch T. Colwell Merit Award in 2012, IEEE VTS 2020 Best Vehicular Electronics Paper Award and over 10 Best Paper Awards from international conferences. Prof. Cao has served as Deputy Editor-in-Chief for IET INTELLIGENT TRANSPORT SYSTEMS JOURNAL, and an Associate Editor for IEEE TRANSACTIONS ON VEHICULAR TECHNOLOGY, IEEE TRANSACTIONS ON INTELLIGENT TRANSPORTATION SYSTEMS, IEEE/ASME TRANSACTIONS ON MECHATRONICS, IEEE TRANSACTIONS ON INDUSTRIAL ELECTRONICS, IEEE/CAA JOURNAL OF AUTOMATICA SINICA, IEEE TRANSACTIONS ON COMPUTATIONAL SOCIAL SYSTEMS, and ASME JOURNAL OF DYNAMIC SYSTEMS, MEASUREMENT AND CONTROL. Prof. Cao is an IEEE VTS Distinguished Lecturer. 
\end{IEEEbiography}

\end{document}